\documentclass[aoas]{imsart}

\RequirePackage{amsthm,amsmath,amsfonts,amssymb}
\RequirePackage[authoryear]{natbib}
\RequirePackage[colorlinks,citecolor=blue,urlcolor=blue]{hyperref}
\RequirePackage{graphicx}
\RequirePackage{booktabs}

\startlocaldefs

\newtheorem{theorem}{Theorem}[section]

\theoremstyle{remark}

\newtheorem{remark}{Remark}

\usepackage{url}
\usepackage{bbm, dsfont,amsthm}
\usepackage{enumerate}
\usepackage{algorithm,algpseudocode} 
\usepackage{color, soul}
\DeclareMathOperator{\E}{\mathbb{E}}
\newcommand{\code}[1]{\texttt{#1}}
\newcommand{\R}{{\normalfont\textsf{R }}{}}
\newcommand{\lsq}{\left[}
\newcommand{\rsq}{\right]}
\newcommand{\lbc}{\left \{ }
\newcommand{\rbc}{\right \} }
\newcommand{\lp}{\left(}
\newcommand{\rp}{\right)}
\newcommand{\cond}{{\, \vert \,}}
\newcommand{\independent}{\perp\!\!\!\perp}
\newcommand{\YCF}{Y^{\text{CF}}}

\endlocaldefs

\begin{document}

\begin{frontmatter}
\title{A flexible sensitivity analysis approach for unmeasured confounding with multiple treatments and a binary outcome with application to SEER-Medicare lung cancer data\thanksref{t1}}
\runtitle{Sensitivity analysis with multiple treatments}
\thankstext{T1}{Liangyuan Hu,  Department of Biostatistics and Epidemiology, Rutgers School of Public Health, Piscataway, NJ 08854, USA. Email:liangyuan.hu@rutgers.edu}

\begin{aug}
\author[A]{\fnms{Liangyuan} \snm{Hu}\ead[label=e1]{liangyuan.hu@rutgers.edu}},
\author[B]{\fnms{Jungang} \snm{Zou}\ead[label=e2]{jz3183@cumc.columbia.edu}},
\author[C]{\fnms{Chenyang} \snm{Gu}\ead[label=e3]{chenyang.gu@analysisgroup.com}},
\author[D]{\fnms{Jiayi} \snm{Ji}\ead[label=e4]{jiayi.ji@mountsinai.org}},
\author[E]{\fnms{Michael} \snm{Lopez}\ead[label=e5]{mlopez1@skidmore.edu}}
\and
\author[F]{\fnms{Minal} \snm{Kale}\ead[label=e6]{minal.kale@mountsinai.org}}
\address[A]{Department of Biostatistics and Epidemiology,
Rutgers University,
\printead{e1}}

\address[B]{Department of Biostatistics,
Columbia University,
\printead{e2}}

\address[C]{Analysis Group, Inc.,
\printead{e3}}

\address[D]{Department of Population Health Science and Policy,
Icahn School of Medicine at Mount Sinai,
\printead{e4}}

\address[E]{Department of Mathematics, Skidmore College,
\printead{e5}}

\address[F]{Department of Medicine,
Icahn School of Medicine at Mount Sinai,
\printead{e6}}

\end{aug}

\begin{abstract}
In the absence of a randomized experiment, a key assumption for drawing causal inference about treatment effects is the ignorable treatment assignment. Violations of the ignorability assumption may lead to biased treatment effect estimates. Sensitivity analysis helps gauge how causal conclusions will be altered in response to the potential magnitude of departure from the ignorability assumption. However, sensitivity analysis approaches for unmeasured confounding in the context of multiple treatments and binary outcomes are scarce.  We propose a flexible Monte Carlo sensitivity analysis approach for causal inference in such settings.  We first derive the general form of the bias introduced by unmeasured confounding, with emphasis on theoretical properties uniquely relevant to multiple treatments. We then propose methods to encode the impact of unmeasured confounding on potential outcomes and adjust the estimates of causal effects in which the presumed unmeasured confounding is removed. Our proposed methods embed nested multiple imputation within the Bayesian framework, which allow for seamless integration of the uncertainty about the values of the sensitivity parameters and the sampling variability, as well as use of the Bayesian Additive Regression Trees for modeling flexibility. Expansive simulations validate our methods and gain insight into sensitivity analysis with multiple treatments. We use the SEER-Medicare data to demonstrate sensitivity analysis using three treatments for early stage non-small cell lung cancer. The methods developed in
this work are readily available in the $\R$ package $\textsf{SAMTx}$.
\end{abstract}


\end{frontmatter}

\section{Introduction}



\label{sec:intro}
\subsection{Overview and objectives} \label{sec:overview}
 In the absence of a randomized experiment, causal inference methods with observational data  can mimic the equivalence between treatment and control groups  to reduce bias due to measured confounders. Demands for comparative effectiveness research involving multiple treatments (i.e., more than two treatment options) have grown substantially. 
 As a motivating example, an important and emerging cancer research question concerns the comparative effectiveness of three commonly used surgical approaches for treating early stage non-small cell lung cancer (NSCLC) tumors. Open thoracotomy (OT) long stood as the standard surgical procedure. With the advent of minimal invasive techniques, video-assisted thoracic surgery (VATS) and robotic-assisted surgery (RAS) have been  increasingly used. However, there is a lack of direct evidence for the comparative effects of these three surgical approaches from head-to-head randomized controlled trials, partially due to difficulty in patient recruitment and high study costs. Large-scale healthcare databases collected in real-world settings are potentially fertile ground for generating the desired evidence. For example, \cite{hu2020estimation} used the Surveillance, Epidemiology, and End Results (SEER)-Medicare data to estimate the average treatment effects of the three surgical approaches on postoperative complications. This work found that compared to OT, VATS led to significantly lower chances of respiratory complication, prolonged length of stay (LOS) (i.e., $>$ 14 days) and intensive care unit (ICU) stay,  but there were no statistically significant differences between RAS and VATS.

Drawing causal inference using observational data, however, inevitably requires  assumptions. A key assumption requires appropriately conditioning on all pre-treatment variables that predict both  treatment and  outcome. The pre-treatment variables are known as confounding variables and this requirement is referred to as the \emph{ignorability} assumption, i.e. no unmeasured confounding \citep{robins1999association, hu2019causal}. This assumption may not be satisfied in real-world data. In our motivating NSCLC example, it has been shown in the literature that pulmonary function is a strong confounder for the treatment effects of the surgical approaches on postoperative complications \citep{ceppa2012thoracoscopic, saito2017impact, VATS5336}. Preoperative pulmonary function test results will guide a clinician to choose an appropriate surgical approach as patients with pulmonary hypertension may have difficulty tolerating changes in venous return with insufflation \citep{VATS5336}; and preoperative pulmonary function also directly predicts the chance of postoperative complications \citep{saito2017impact}. It is also suggested that preoperative physical activity level is  a likely confounder as it affects a clinician's choice of surgical approaches and predicts surgical outcomes following lung cancer resection  \citep{bille2021preoperative}. The SEER-Medicare data set does not include information on preoperative pulmonary function or physical activity level, therefore these two variables are unmeasured confounders and the ignorability assumption is violated in the NSCLC study using SEER-Medicare data.

Appropriate techniques are needed to handle the potential magnitude of departure from the ignorability assumption. One recommended approach is sensitivity analysis  \citep{von2007strengthening}. However, the sensitivity analysis methods are underdeveloped in the context of multiple treatments and binary outcomes. In this article, we propose a flexible sensitivity analysis approach for such settings. We derive the form of the bias in the causal effect estimate when there exists unmeasured confounding and shed light on the bias composition  uniquely pertaining to multiple treatments.  We formulate the impact of unmeasured confounding directly in respect of the potential outcomes, and propose and apply strategies to posit the plausible degrees of impact. We then 
construct the ``corrected'' causal effect estimators using the Bayesian modeling framework to account for the uncertainty about the surmised impact of unmeasured confounding. We conduct a large-scale simulation study to examine our proposed methods. A comprehensive case study applies our methods to the SEER-Medicare data to elucidate how causal conclusions would  change to various degrees of unmeasured confounding about the effects of the three surgical approaches (RAS vs. OT vs. VATS) on four postoperative complication outcomes: respiratory complication, prolonged LOS, ICU stay and 30-day readmission.

This paper is organized as follows: the remainder of Section~\ref{sec:intro} provides additional background of the ignorability assumption and broadly reviews sensitivity analysis approaches. Section~\ref{sec:SA} describes notation, defines a bias formula and proposes methods to study the sensitivity of causal effect estimates to unmeasured confounding. Section~\ref{sec:sim} develops a wide variety of simulation scenarios to examine the operating characteristics of our proposed methods, and presents findings. In Section~\ref{sec:application}, we apply our approach to SEER-Medicare data to study how sensitive the causal effect estimates of the three surgical procedures are to different levels of unmeasured confounding. Section~\ref{sec:discussion} concludes with a discussion. 

\subsection{The ignorability assumption for causal inference using observational data}

In the absence of a randomized experiment, the ignorability assumption is needed for the causal estimand to be identifiable. This assumption, however, can be violated in observational studies. When the treatment assignment is not ignorable, treatment effect estimates may be biased. 
Many approaches have been proposed to weaken the reliance on the ignorability assumption in observational studies by using quasi-experimental designs and natural experiments.  However, these methods come with their own sets of assumptions (e.g. ignorability of the instrument for the instrumental variables approach),  and violations of these assumptions can lead to biased effect estimates. In addition, it is arguably difficult to find data that meet the criteria laid out by these methods while addressing the research question of primary interest \citep{dorie2016flexible}. 

One widely recognized way to address concerns about violations of ignorability is via a sensitivity analysis.  In fact, the Strengthening the Reporting of Observational Studies in Epidemiology (STROBE) guidelines recommend research in observational settings be accompanied by sensitivity analysis  investigating the impact of potential unmeasured confounding \citep{von2007strengthening}. Sensitivity analysis is important to elucidating the ramifications of the ignorability assumption by assessing the degree to which unmeasured confounding can alter the causal conclusions.

\subsection{Broad overview of existing sensitivity analysis approaches}

In a sensitivity analysis, the effects of unmeasured confounding are encoded in one or more unidentifiable numerical parameters -- commonly referred to as sensitivity parameters \citep{gustafson2018}. In a sensitivity analysis, we assume a plausible set of values for the sensitivity parameters and combine them with the information observed in the data to provide an adjusted inference about a target parameter. 

One common approach models the effect of unmeasured confounders $U$ on causal conclusions. In such a sensitivity analysis, one introduces $U$ respectively into the outcome model and the treatment assignment model, and varies the association  of $U$, via regression coefficients (sensitivity parameters), with the outcome and the treatment. How the causal effect estimates change with different values of sensitivity parameters quantifies the sensitivity of the causal inferences about treatment effects to the potential magnitude of departure from the ignorability assumption. This approach is referred to as the ``tabular'' method  \citep{gustafson2018}  or the external adjustment approach \citep{kasza2017assessing}. 
There is  considerable disagreement on the modality of $U$, e.g., whether to view $U$ as continuous or binary \citep{rosenbaum1983assessing,hu2017modeling}, or whether $U$ is univariate or multidimensional \citep{lin1998assessing, imbens2003sensitivity}, or whether there is no interaction between the effects of $U$ and the treatment on the outcome \citep{rosenbaum1983assessing, imbens2003sensitivity}. 
A more general external adjustment approach allows $U$ to be multidimensional and of various data types, and does not assume independence between $U$ and measured confounders \citep{ding2016sensitivity}. However, in exchange for the relaxation of assumptions, a large number of parameters must be considered across strata of measured confounders, sacrificing interpretability.

Bayesian approaches have also been proposed in the vein of external adjustment via $U$.  A joint prior distribution is placed over the sensitivity parameters to model beliefs about a possible unmeasured confounding mechanism. The Bayesian updating is applied via Markov chain Monte Carlo (MCMC), to combine the effects of prior distributions and the data. The advantage of a Bayesian sensitivity analysis is that the posterior inference about a target parameter incorporates uncertainty about the values of the sensitivitiy parameters. Examples of Bayesian sensitivity analysis approaches include \cite{daniels2008missing,  gustafson2018, hogan2014bayesian}.  A ``near-Bayesian'' approach, Monte Carlo sensitivity analysis,  has attracted a lot of attention in recent years, for it is intuitive to understand and does not require extensive Bayesian computation. 
Examples of this approach include \cite{greenland2005multiple, lash2011applying, mccandless2017comparison}.
Neither of the two probabilistic sensitivity analysis methods forgoes the introduction of a hypothetical unmeasured confounder $U$, and consequently the issues around assumptions about the underlying structure of $U$ remain.  


Another body of sensitivity analysis literature fits into the \emph{confounding function} paradigm. Pioneered by \cite{robins1999association}, this approach directly 
describes the net confounding on the mean potential outcomes without any reference to $U$, using the between-group difference in the average potential outcomes within levels of measured confounders.  The confounding function approach is related to the external adjustment approach in that the potential outcomes can be considered as the ultimate unmeasured confounder $U$.  The known information of the potential outcomes and treatment renders other unmeasured confounders superfluous because the observed outcome is a deterministic function of the treatment and the potential outcome \citep{brumback2004sensitivity}.  This approach has gained wide attention, particularly in the epidemiology literature \citep{brumback2004sensitivity, li2011propensity, kasza2017assessing}.  By contrast with the ``tabular'' method where a hypothetical $U$ is introduced into both the treatment assignment and outcome models, this approach is preferred when there is no sufficient domain knowledge to inform the magnitude of association of $U$ with both treatment and outcome,  and  when the primary interest is in understanding the entirety of the effect of \emph{all} unmeasured confounding.


Existing sensitivity analysis approaches have largely focused on a \emph{binary} treatment.  \cite{vanderweele2011unmeasured} derived a general class of bias formulas for categorical or continuous treatments. However, in order to achieve interpretability,  a binary $U$ needs to be hypothesized with a constant effect on the outcome across treatment and covariate levels.
Moreover, as a ``tabular'' method,  the uncertainty about the sensitivity parameters is not incorporated into the sensitivity analysis.  

\section{Sensitivity analysis for unmeasured confounding}\label{sec:SA}

\subsection{Notation, definitions and assumptions}
We base our approach on the potential outcomes or counterfactural framework \citep{rubin1974estimating}.  Consider a sample of $N$ units, indexed by $i=1, \ldots, N$, drawn randomly from a target population, and the causal effect of treatment $A$ on a binary outcome $Y$, where $A$ takes values $a$ from a set of $J$ possible treatments indexed by $\mathcal{A} = \{a_1, a_2,  \ldots, a_J\}$, and $Y=1$ indicates events.  The number of units receiving treatment $a_j$ is $n_{a_j}$, where $\sum_{j=1}^J  n_{a_j}=N$. For each individual $i$, there is a vector of pre-treatment measured covariates, $X_i$.  Let $Y_i$ be the observed outcome of the $i$th individual and $\{Y_i(a_1), \ldots, Y_i(a_J)\}$ the potential outcomes under each treatment option of $\mathcal{A}$. The generalized propensity scores for treatment assignment sum to 1, $\sum_{j=1}^J \mathbb{P}(A_i=a_j \cond X_i) =1$. For each individual, at most one of the potential outcomes is observed, i.e., the observed outcome is equal to the potential outcome when the treatment had been set to what  in fact was for the individual. Formally, $Y_i = \sum_{j=1}^J Y_i(a_j) \mathbbm{1}(A_i=a_j)$. 

The  standard identifying assumptions for causal inference with observational data are:
\begin{itemize}
\item [(A1)] SUTVA: The stable unit treatment value assumption. The set of potential outcomes for an individual does not depend on the treatments received by other individuals. 
\item[(A2)] Overlap: The generalized propensity score for treatment assignment is bounded away from 0 and 1. That is, $0<\mathbb{P}(A =a\cond X) <1$. 
\item[(A3)]  Ignorability:  Treatment assignment is conditionally independent of the
potential outcomes given the covariates, $Y(a) \independent A \cond X$. 
\end{itemize}
In this paper,  we deal with the situation where assumption (A3) breaks, while maintaining assumptions (A1) and (A2). 
The causal effects are summarized by estimands, broadly defined as the contrast between functionals of the individual-level potential outcomes on a common set of individuals.  For binary outcomes, common causal estimands are the risk difference, relative risk or odds ratio. We focus on the risk difference  because it has clear interpretability and works well with additive confounding functions (see a brief discussion in Section~\ref{sec:MCSA}). 

Causal estimands of primary interests are average treatment effects defined either over the sample or the population.  Consider a pairwise comparison between treatments $a_j$ and $a_k$ in terms of the risk difference. Common sample estimands are the sample average treatment effect (SATE), $\frac{1}{N} \sum_{i=1}^N [ Y_i(a_j) - Y_i(a_k) ]$. Common population estimands are the population average treatment effect (PATE), $\E[Y(a_j) -Y(a_k)]$. Conditional average treatment effect (CATE) $\frac{1}{N}\sum_{i=1}^N \E \lsq Y_i(a_j)-Y_i(a_k) \cond X_i \rsq$ is another estimand that preserves some of the properties of the previous two \citep{hill2011bayesian}. As our approach is embedded in a Bayesian framework, CATE is a natural estimand to use in this paper \citep{hill2011bayesian}.  By averaging the individual conditional expectation $\E[Y_i(a_j)-Y_i(a_k) \cond X_i]$ across the empirical distribution of $\{X_i\}_{i=1}^{N}$, we obtain the sample marginal effects \citep{hu2020estimation}. Researchers sometimes are also interested in estimating the average treatment effect on the treated (ATT). All the three estimands described above have their ATT counterparts. For purposes of illustration, we focus on CATE in this paper, but our method can be straightforwardly extended for the ATT, by averaging the differenced potential outcomes over those in the reference group. For example, the conditional average treatment effect among those who received treatment $a_j$ $\text{CATT}_{a_j \cond a_j, a_k}$ is $\frac{1}{N_j}\sum_{i:A_i=a_j} \E \lsq Y_i(a_j)-Y_i(a_k) \cond X_i \rsq$, where $N_j = \sum_{i=1}^N I(A_i =a_j)$ is the size of the reference group $a_j$.

\subsection{Bias formulas for average treatment effects with multiple treatments}
Following \cite{robins1999association},  we define the confounding function for any pair of treatments $(a_j, a_k)$ as
\begin{eqnarray} \label{eq:cf}
c(a_j, a_k, x) = \E \lsq Y(a_j) \cond A = a_j, X=x\rsq - \E \lsq Y (a_j) \cond A =a_k, X=x \rsq.
\end{eqnarray} 
The  confounding function (sensitivity parameter) directly represents the difference in the mean potential outcomes $Y(a_j)$ between those treated with $A=a_j$ and those treated with $A=a_k$, who have the same level of $x$.  Under the ignorability assumption, given the measured confounders $X$, the potential outcome is independent of treatment assignment. That is, the individuals who received treatment $A=a_j$ and the individuals who received treatment $A=a_k$ are conditionally exchangeable given measured $X=x$. Had they received the same treatment $a_j$, their mean potential outcomes would have been the same, or $c(a_j, a_k, x) = 0$, $\forall \{a_j, a_k\} \in \mathcal{A}$  ($a_j \neq a_k$ by default). When this assumption is violated and the unmeasured confounding is present, the causal effect estimates will be biased. Theorem~\ref{thm:biasform} shows the general form of the bias in the estimated  causal effects in the multiple treatment setting when treatment assignment is not ignorable. A proof of the theorem is presented in the \hyperref[appn]{Appendix}.
\begin{theorem}\label{thm:biasform}
Consider the pairwise treatment effect between $a_j$ and $a_k \in \mathcal{A} = \{a_1, \ldots, a_J\}$ on the basis of the relative risk.  If $Y(a)  \not\!\perp\!\!\!\perp A |X$, then ignoring unmeasured confounding will give rise to a biased estimate of the causal effect. The bias formula is 
\begin{align} \label{eq:biasform}
\begin{split}
\text{Bias}(a_j,a_k) =&-p_{j} c(a_k, a_j, x) + p_{k}c(a_j,a_k,x)\\
&-\sum\limits_{l: l \in \mathcal{A}\setminus\{a_j, a_k\}} p_{l} \lbc c(a_k, a_l, x) -c(a_j,a_l,x) \rbc, 
\end{split}
\end{align}
where $p_{j} = \mathbb{P}(A= a_j \cond X= x)$, $j \neq k \in \{1, \ldots, J\} $.
\end{theorem}

Theorem~\ref{thm:biasform} provides a key apparatus for our sensitivity analysis approach. The general bias formula provides  several insights into the estimates of causal effects in the multiple treatment setting. 
\begin{remark}
The bias takes the form of a linear combination of the confounding functions. When the ignorability assumption holds, all the confounding functions are equal to zero, and consequentially there will be no bias in the estimated causal effect. 
\end{remark}

\begin{remark}
The bias depends not only on the pair of treatments of interest $\{a_j, a_k\}$, but also on contrasts between them and all other treatment options. This is a property uniquely pertained to the multiple treatment setting. Note that this is true for both the average treatment effect and the average treatment effect on the treated, as the bias arising from unmeasured confounding is derived at the individual level.
\end{remark}

\begin{remark}
A key component of the bias is the generalized propensity score. This is closely related to the findings in recent causal inference literature  that inclusion of the propensity score in the outcome model formulation could help achieve better bias reduction in the estimated causal effect.  The role of the propensity score becomes more important with strong targeted selection \citep{hahn2020bayesian, hucomment2020}. 
\end{remark}

\subsection{Confounding function adjusted causal effect estimates}           
We propose a flexible Monte Carlo sensitivity analysis approach \citep{mccandless2017comparison} to construct confounding function adjusted causal effect estimators \citep{li2011propensity}.  We first investigate strategies of formulating the confounding function $c$ in equation~\eqref{eq:biasform} to posit the level of unmeasured confounding. We then propose an outcome modeling based approach to adjust the causal effect estimates in which the bias in equation~\eqref{eq:biasform} with presumed $c$ has been ``corrected''. We  leverage Bayesian Additive Regression Trees (BART) for outcome modeling.  BART has gained reputation for its prediction accuracy while still having regularized model for preventing overfitting and the
minimum of tuning \citep{chipman2010bart,hu2020tree,hu2020ranking,hu2021variable}.  When adapted into causal inference, numerous studies have shown that BART based approaches can yield more accurate effect estimates and provide coherent interval estimates  \citep{hill2011bayesian, hu2020estimation,  hu2020rare,hu2021estimatinghet,hu2021estimatingaoe}. For the adjusted inference about treatment effects accounting for unmeasured confounding, We embed nested multiple imputation (MI)  in the Bayesian framework. The combined adjusted causal effects and uncertainty intervals can be estimated by Rubin's formula for nested MI \citep{rubin2003nested, gu2019development}. However, as an anonymous reviewer pointed out, when presuming normality of the posterior distribution is not justifiable, a more appropriate way for inference is by pooling posterior samples across model fits arising from the multiple data sets \citep{zhou2010note}. We adopt this strategy to implement our  sensitivity analysis approach so that the concern as to whether the normality assumption holds is done away with.

\subsubsection{Monte Carlo sensitivity analysis}\label{sec:MCSA}
Given known confounding functions $c$, we can construct the confounding function adjusted estimators by modifying the observed outcome $Y$ and estimating the causal effect via outcome modeling using the adjusted outcome $\YCF$. We will discuss strategies for formulating the $c$ function in Section~\ref{sec: interp}. Theorem~\ref{thm: Ycorrection} proposes an approach to computing $\YCF$.  A proof of the theorem is presented in the \hyperref[appn]{Appendix}.

\begin{theorem} \label{thm: Ycorrection}
Under assumptions (A1) and (A2), the estimation of the causal effect based on the adjusted outcome $\YCF$  effectively removes the bias in equation~\eqref{eq:biasform}. Suppose individual $i$ received treatment $a_j \in \mathcal{A} =\{a_1, \ldots, a_J\}$, then the adjusted outcome for individual $i$  is defined as 
\begin{eqnarray}\label{eq:Ycorrection}
\YCF_i \equiv Y_i - \sum_{l \neq j}^J \mathbb{P}(A_i = a_l\cond X_i=x) c(a_j, a_l,x).
\end{eqnarray}
\end{theorem}

For a binary $Y$, we compute $\YCF$ as in equation~\eqref{eq:Ycorrection} and treat $\YCF$ as continuous for outcome modeling. This is amenable to the causal estimand  based on the risk difference. For ratio-based estimands such as the relative risk, one option is to define the confounding function as $c(a_j,a_k,x) = \log \lbc\frac{\E \lsq Y(a_j)  \cond A=a_j,x \rsq}{\E \lsq Y(a_k) \cond A=a_k,x \rsq}\rbc$, and use a multiplicative correction in the form of $Y   \frac{\E \lsq Y(a_j) \cond A = a_j, x \rsq} {\E \lsq Y (a_k) \cond A=a_k, x \rsq}$. However, this formulation leads to a less clear interpretability of the confounding function.  Moreover, the adjusted outcomes will be 
a mixture of zeros and nonnegative continuous data, posing challenges for outcome modeling. In this paper, we focus on the additive confounding functions. 

\begin{algorithm}[H] 
\caption{Monte Carlo Sensitivity Analysis} \label{alg:MCSA}
\begin{enumerate}
\item Fit a multinomial probit BART model $f^{\text{MBART}}(A\cond X)$ to estimate the generalized propensity scores, $p_l \equiv \mathbb{P} (A =a_l \cond X =x) \;  \forall a_l \in \mathcal{A}$, for each individual. 
\item  
\begin{algorithmic}
\For{$j \gets 1$ to $J$}                    
 \State {Draw $M_1$ generalized propensity scores $\tilde{p}_{l1}, \ldots, \tilde{p}_{lM_1}, \forall l \neq j \wedge a_l \in \mathcal{A}$ from the posterior predictive distribution of $f^{\text{MBART}} (A \cond X)$ for each individual.} 
    \For{$m \gets 1$ to $M_1$}                    
        \State {Draw $M_2$ values $\gamma^*_{lm1}, \ldots, \gamma^*_{lmM_2}$ from the prior distribution of each of the confounding functions $c(a_j, a_l, x)$, for each $l \neq j \wedge a_l \in \mathcal{A}$. } 
    \EndFor
    \EndFor
\end{algorithmic}
\item Compute the adjusted outcomes $\YCF$, for each treatment $a_j$, as in equation~\eqref{eq:Ycorrection} of Theorem~\ref{thm: Ycorrection} for each of $M_1M_2$ draws of $\{\tilde{p}_{l1},  \gamma^*_{l11}, \ldots, \gamma^*_{l1M_2}, \ldots,  \tilde{p}_{lM_1},  \gamma^*_{lM_11}, \ldots,  \gamma^*_{lM_1M_2};  l \neq j \wedge a_l \in \mathcal{A}\}$. 
\item Fit a BART model to each of $M_1\times M_2$ sets of observed data with  the adjusted outcomes $\YCF$, and estimate the combined adjusted causal effects and uncertainty intervals by pooling posterior samples across model fits arising from the $M_1 \times M_2$ data sets.    
\end{enumerate}
\end{algorithm}

We propose a Monte Carlo sensitivity analysis in Algorithm~\ref{alg:MCSA}. Nested MI is used to draw samples for the product term $\mathbb{P}(A_i = a_l\cond X_i=x) c(a_j, a_l,x)$ in equation~\eqref{eq:Ycorrection}.  
The generalized propensity scores are drawn from a multinomial probit BART model, and the $c$ values are drawn from user-specified functions. The $M_1 \times M_2$ nested MI are used to construct $M_1 \times M_2$ complete data sets.  
For the estimation of the causal effect, we fit a BART model to each of the $M_1 \times M_2$ data sets with the adjusted outcome $\YCF$.  In the multiple treatment setting, the quantity of interest is the pairwise treatment effect, $Q_{a_j,a_k} = \text{CATE}_{a_j, a_k}, \forall \{a_j, a_k\} \in \mathcal{A}$. Following the estimation procedure described in \cite{hu2020estimation}, we compute, for each data set, the posterior mean causal effect $Q^{m_1,m_2}_{a_j, a_k} = \widehat{\text{CATE}}_{a_j, a_k}^{m_1,m_2}, \forall \{a_j, a_k\} \in \mathcal{A}$,
$m_1 =1, \ldots, M_1, m_2 =1, \ldots, M_2$. The overall estimates of the causal effect $\bar{Q}_{a_j, a_k}$ and sampling variance $\bar{U}_{a_j, a_k}$ are obtained as the posterior mean and variance of the pooled posterior samples $\{Q^{m_1,m_2}_{a_j, a_k}: a_j, a_k \in \mathcal{A}, m_1 =1,\ldots, M_1, m_2=1,\ldots, M_2 \}$ across the $M_1 \times M_2$ data sets \citep{zhou2010note}. When implementing our approach in simulations (Section~\ref{sec:sim}), we used $M_1=10$ and $M_2=10$ for the large sample size $N=10,000$ and $M_1=30$ and $M_2=30$ for the smaller sample size $N=1500$. In the case study (Section~\ref{sec:application}), we used $M_1=30$ and $M_2=30$.

\subsubsection{Priors for sensitivity parameters}\label{sec: interp}
 There are three ways in which we can specify the prior for the confounding functions $c(\cdot)$: (i) point mass prior; (ii) re-analysis over a range of point mass priors (tipping point); (iii) full prior with uncertainty specified.  Our approach uses strategy (iii), offering an advantage of allowing the incorporation of uncertainty about unidentified components ($c$ functions) of the model formally into the analysis. 
 To  surmise a reasonable prior for $c(\cdot)$, we follow  three guidelines. 
 \begin{enumerate}[(1)]
 \item  Subject-matter expertise should  be leveraged to assume a plausible distribution for the $c(\cdot)$ to encode our prior beliefs about the possible direction and magnitude of the effect of unmeasured confounding.  Following \cite{robins1999association} and \cite{brumback2004sensitivity}, $c(\cdot)$ can be specified as a scalar parameter or as a functional of the measured covariates $X$. 
 \item With a binary outcome, the $c(\cdot)$ describes the difference in two probabilities, providing the  \emph{natural} bounds of $[-1,1]$.
 \item The bounds of the prior for $c(\cdot)$ can be further reduced by assuming that  the entirety of unmeasured confounding accounts for less than a certain units of the remaining standard deviation unexplained by $X$ \citep{hogan2014bayesian}.  
 \end{enumerate}
 
 Using our motivating NSCLC example (Section~\ref{sec:overview}), three treatment options are $A=1$: RAS, $A=2$: OT and $A=3$: VATS. The outcome events are postoperative respiratory complications. Consider a scalar parameter for $c(1,3,x)$ and $c(3,1,x)$ comparing RAS with VATS. We assume the unmeasured factors (e.g. lung function test results and physical activity level) guiding clinicians to make a treatment decision tend to lead them to recommend RAS over VATS to healthier patients. This assumption is based on prior domain knowledge, which suggests that a patient who is high risk for OT is also high risk for RAS \citep{VATS5336}, and that VATS will remain as the gold standard for lung cancer surgery for its proven safety and feasibility \citep{sihoe2020video}. This assumption implies $c(1,3,x) <0$ and $c(3,1,x) >0$, i.e., relative to those treated with VATS,  patients who were assigned to RAS, on average, will have lower potential likelihood of experiencing postoperative respiratory complications to both RAS and VATS. This establishes the  upper (lower) bound of the range of values for the sensitivity parameter $c(1,3,x)$ ($c(3,1,x)$). For the other bound,  we first estimate the remaining standard deviation unexplained by $X$ (treating binary $Y$ as continuous), $\hat{\sigma}^2$. Then assume that $c(1,3,x) > -h\hat{\sigma}$ and $c(3,1,x) < h\hat{\sigma}$, i.e.,  unmeasured confounding would account for less than $h$ units of the remaining standard deviation unexplained by measured confounders $X$. In conjunction with the natural bounds of $[-1,1]$, we can assume the support of the distribution is $[0, \min (h\hat{\sigma}, 1) ]$ for $c(1,3,x)$ and $[\max (-h\hat{\sigma}, -1 ), 0 ]$ for $c(3,1,x)$. Other plausible assumptions about the $c(\cdot)$'s and their interpretations are presented in Table~\ref{tab:cfun-inter}. Section~\ref{sec:application} describes detailed considerations for specifying plausible priors for the confounding functions when conducting a sensitivity analysis  using the SEER-Medicare NSCLC data.
 
 \begin{table}[H]
\centering
\caption{Interpretation of assumed priors on $c(a_1, a_2,x)$ and $c(a_2,a_1,x)$ for causal estimands based on the risk difference. Postoperative complications are used as an example for the outcome.}
\label{tab:cfun-inter}
\begin{tabular}{ccp{0.8\textwidth}}
\toprule
\multicolumn{2}{c}{Prior assumption}  & Interpretation and implications of the assumptions\\
$c(a_1,a_2,x)$ &$c(a_2,a_1,x)$ &\\\hline
$>0$ &$<0$ &  Individuals treated with $a_1$ will have higher potential probability of experiencing postoperative complications to both $a_1$ and $a_2$  than individuals treated with $a_2$; i.e. unhealthier individuals are treated with $a_1$. \\
$<0$ &$>0$ & Contrary to the above interpretation, healthier individuals are treated with $a_1$. \\
$<0$ &$<0$ & The potential postoperative complication probability to $a_1$ ($a_2$) is lower among those who choose it than among those who choose $a_2$ ($a_1$).  Thus, the observed treatment allocation between these two approaches is beneficial relative to the alternative which reverses treatment assignment for everyone. \\
$>0$ &$>0$ & Contrary to the above interpretation, the observed treatment allocation between these two approaches is undesirable relative to the alternative which reverses treatment assignment for everyone.\\ 
\bottomrule
\end{tabular}
\end{table}

\section{Simulation study}\label{sec:sim}
We conduct a series of simulations to evaluate the operating characteristics of our proposed sensitivity analysis approach. The data generating processes cover a wide variety of scenarios motivated by the data structures of the SEER-Medicare NSCLC registry. We first carry out an illustrative simulation to empirically verify our sensitivity analysis approach and demonstrate  properties uniquely relevant to  multiple treatments. In the second set of simulations, we contextualize  sensitivity analysis  in the multiple treatment settings with varying sample sizes,  degrees of unmeasured confounding and sparsity levels of covariate overlap. Three treatment levels $\mathcal{A} = \{1, 2, 3\}$ are used throughout. The estimand is   CATE based on the risk difference. The formula~\eqref{eq:Ycorrection} in Theorem~\ref{thm: Ycorrection} reduces to 
\begin{equation} \label{eq:YCF-sim}
 \YCF= \begin{cases}
      Y - p_2 c(1,2,x) -p_3 c(1,3,x) \hspace{5mm} \text{if } A =1 \\
      Y - p_1 c(2,1,x) - p_3 c(2,3,x)\hspace{5mm} \text{if } A =2 \\
      Y - p_1 c(3,1,x) - p_2 c(3,2,x)  \hspace{5mm} \text{if } A =3 \\
            \end{cases}.
\end{equation}
Note that for each pairwise treatment effect,  four confounding functions are involved. We replicate each simulation scenario 1000 times.  To judge the appropriateness of our sensitivity analysis approach, we use the bias and root mean squared error (RMSE). In addition, we examine the coverage probabilities of our confounding-function-adjusted causal effect estimates under three levels of  covariate overlap: strong, moderate and weak overlap.

\subsection{An illustrative simulation}\label{sec:illus-sim}
We considered a total sample size $N=1500$ with a balanced treatment allocation (the ratio of units = 1:1:1) across three treatment groups, a binary confounder $X_1 \sim \text{Bernoulli} (0.4)$, and a binary unmeasured confounder $U \sim \text{Bernoulli} (0.5)$. Both treatment assignment and outcome generating mechanisms depend on $X_1$ and $U$, but only $X_1$ is observed. The treatment assignment followed a multinomial distribution 
$A \cond X_1, U \sim \text{Multinomial} (N, \pi_1, \pi_2, \pi_3),$ where $\pi_{1}=\frac{e^{l_1}}{e^{l_1}+e^{l_2}+e^{l_3}},\pi_{2}=\frac{e^{l_2}}{e^{l_1}+e^{l_2}+e^{l_3}},\pi_{3}=\frac{e^{l_3}}{e^{l_1}+e^{l_2}+e^{l_3}}$, $l_{1}  =0.2X_1+0.4U+N(0,0.1)$, $l_{2}  =-0.3X_1+0.8U+N(0,0.1)$, and $l_{3}  =0.1X_1+0.5U+N(0,0.1)$.  The $N(0,0.1)$ was used to generate measurement errors, and the coefficients of $X_1$ and $U$ were chosen to create a balanced treatment allocation as well as reasonably good covariate overlap \citep{hu2020estimation}. We generated three sets of nonparallel response surfaces -- models for the outcome (or potential outcomes) conditional on the treatment
and confounding covariates \citep{hill2011bayesian},
\begin{align*}
\mathbb{P} \lp Y (1) =1 \mid X_1, U \rp& =\text{logit}^{-1}(-0.8X_1-1.2U+1.5X_1U)\\
\mathbb{P} \lp Y (2) =1 \mid X_1, U \rp & =\text{logit}^{-1}(-0.6X_1+0.5U+0.3X_1U)\\
\mathbb{P} \lp Y (3) =1 \mid X_1, U \rp & =\text{logit}^{-1}(0.3X_1+1.3U+0.2X_1U).
\end{align*}
The parameters were chosen to induce similar outcome event probabilities as the probabilities of respiratory complication and  ICU stay observed in three treatment groups of the SEER-Medicare data (Section~\ref{sec:application}). We demonstrate our sensitivity analysis approach in a realistic scenario where the unmeasured confounder and the measured confounder may have an interaction effect on the outcome. Under this simulation configuration, the observed outcome event probability was 0.40 in $A=1$, 0.51 in $A=2$ and 0.64 in $A=3$ and the true $\text{CATE}_{1,2}=-0.16$, $\text{CATE}_{1,3}=-0.29$ and $\text{CATE}_{2,3}=-0.13$. In Web Figure 1, we show that all results and conclusions hold for a simplified situation where there does not exist an interaction between $X_1$ and $U$. 

 With three treatments, there are six sensitivity parameters, $c(1, 2,x_1)$, $c(1,3,x_1)$, $c(2,1,x_1)$, $c(2,3,x_1)$, $c(3,1,x_1)$ and $c(3,2,x_1)$.   Suppose we do not get to observe $U$, from the simulated data we know, for any $a_1 \neq a_2 \in \{1,2,3\}$, the true $c(a_1, a_2,x_1) =\E [Y(a_1) \cond A =a_1, X_1=x_1] - \E [Y(a_1) \cond A =a_2, X_1=x_1] \; \forall x_1 \in \{0,1\}$. Ignoring effect modification by $X_1$, we can calculate the true $c^0(a_1,a_2) = \E [Y(a_1) \cond A =a_1] - \E [Y(a_1) \cond A =a_2] $. The residual standard deviation of the model including measured covariate $X_1$ is $\hat{\sigma}=0.48$. 
When implementing our sensitivity analysis approach, we formulated the confounding function $c(\cdot)$ in two ways. In one way,  $c(\cdot)$ was specified as a scalar parameter and in another way, $c(\cdot)$ was described within each stratum of $X$. In each consideration for the functional form of $c(\cdot)$, we investigated the performance of our sensitivity analysis approach using the following four strategies for specifying the prior distribution on $c(\cdot)$: 
\begin{enumerate}[(I)]
\item The true confounding function $c^0$. This is intended for empirically verifying  our proposed sensitivity analysis approach.  
\item A uniform distribution on the interval centered around $c^0$, $\mathcal{U} \lp \max(-1, c^0 - h\hat{\sigma}), \min(1, c^0+h\hat{\sigma}) \rp$.
\item Shifting the uniform distribution in (II) away from the truth $c^0$, $\mathcal{U} \lp \max(-1, c^0 - 2h\hat{\sigma}),  c^0 \rp$, or $\mathcal{U} \lp c^0, \min(1, c^0+2h\hat{\sigma}) \rp$. We refer to this setup as \emph{moving the goal posts}. 
\item A uniform distribution with the natural bounds $\mathcal{U}(-1,1)$. This case makes the most noninformative assumptions about the sensitivity parameters. 
\end{enumerate}

In cases (II)-(IV), other distributions can be assumed in place of the uniform distribution. We argue that for a sensitivity analysis, the range of plausible values for the sensitivity parameters is more important than the distributional shape of the sensitivity parameters. We show in Web Figure 2 that for the sensitivity parameters, the truncated normal distribution yielded similar sensitivity analysis results as the uniform distribution.

Figure~\ref{fig: binary-sim} displays the estimates of three pairwise causal  effects among 1000 replications corresponding to each of four strategies (I)-(IV). Along with the sensitivity analysis results, estimates ignoring $U$ and including $U$ are also presented. In case (I) where the true $c(\cdot, \cdot, x_1)$ functions were specified, our sensitivity analysis approach yielded near-zero biases in the CATE estimates, similar to the results that could be achieved if $U$ were actually observed.  These results provide empirical evidence for the validity of our sensitivity analysis approach. Using a scalar parameter for the $c$ function moderately increased the biases. 
For exposition purposes, we considered a scenario \emph{``I: 3rd $A$ ignored''} , in which only the $c$ functions involving the pair of treatments for the pairwise CATE of interest were considered, while the other two $c$ functions involving the third treatment were ignored. For example, suppose we were interested in estimating $CATE_{1,2}$. Using this strategy,  $c(1,2,x_1)$ and $c(2,1,x_1)$ were considered while $c(1,3, x_1)$ and $c(2,3,x_1)$ involving $A=3$ were set to zero. The results show that ignoring the third treatment that is not in the target pairwise CATE, even if we use the true confounding functions for the other two treatments, the adjusted CATE estimates would still be biased. The finding attests to Theorem~\ref{thm:biasform} and highlights the importance of simultaneous consideration of all treatments  for causal inference with multiple treatments. Web Table 2 summarizes the average absolute bias and RMSE in each of the estimates of treatment effects for each sensitivity analysis strategy, conveying the same messages as Figure~\ref{fig: binary-sim}. In Web Figure 3, we present the credible intervals of the three pairwise causal effects for each scenario considered using a random data replication. Results from Figure~\ref{fig: binary-sim} and Web Figure 3 are congruent. Under the correct specification of the sensitivity parameters (case I), 
the width of credible intervals for the adjusted CATE is similar to that for the CATE had $U$ been observed. As the uncertainty about the sensitivity parameters grows (Cases II and III), the width of credible intervals also increases.

\begin{figure}[H] 
\includegraphics[width=1\textwidth]{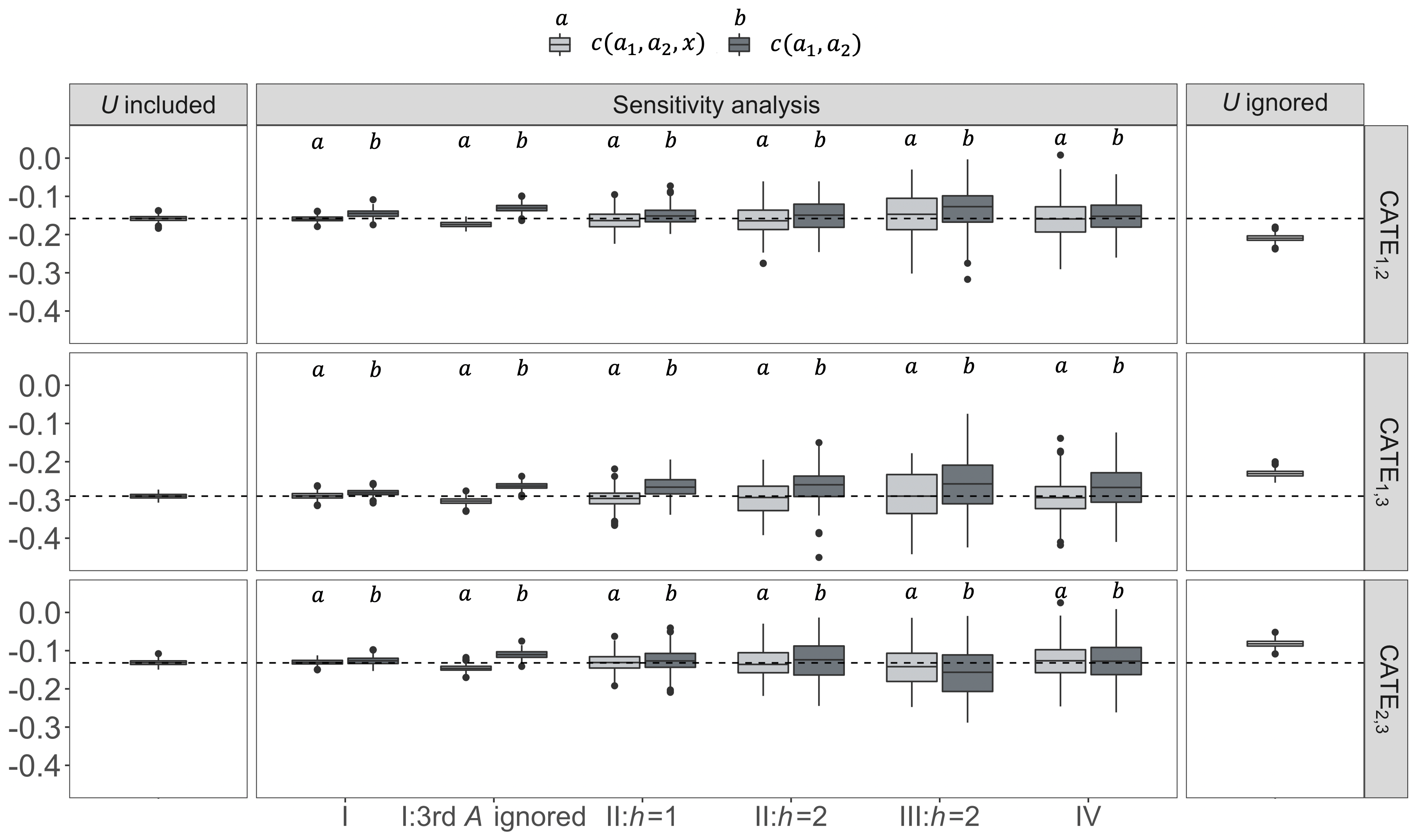}
\caption{Estimates of three pairwise causal effects $\text{CATE}_{1,2}$, $\text{CATE}_{1,3}$ and $\text{CATE}_{2,3}$ among 1000 replications. For sensitivity analysis, strategies (I)-(IV) described in Section~\ref{sec:illus-sim} were used to specify the prior distributions for the confounding functions $c(\cdot)$. For strategy (I),  the scenario ``3rd $A$ ignored'' considers only the $c(\cdot)$ functions involving the pair of treatments for the target CATE, while setting the  $c(\cdot)$ functions involving the third treatment to zero. For strategy (II), both $h=1$ and $h=2$ are considered, representing one and two remaining standard deviation, respectively. The CATE results that could be achieved if $U$ were actually observed and the naive CATE estimators ignoring $U$  are also presented. Red dashed lines mark the true CATE.} 
\label{fig: binary-sim}
\end{figure}

\subsection{Contextualizing simulation in multiple treatment settings}
\subsubsection{Simulation design}\label{sec:complex-sim-design}
We contextualize a set of simulations in the multiple treatment settings, motivated by the SEER-Medicare data used in Section~\ref{sec:application} , with three design factors. The first factor concerns the sample size and ratio of units across treatment groups. We considered (i) $N=1500$ with a 1:1:1 ratio, and (ii) $N=10,000$ with a 1:10:9 ratio, which mimics the sample composition of the three treatment options in the SEER-Medicare NSCLC registry. The second factor is the degree of unmeasured confounding, for which we considered three levels covering low, mid and high magnitude of departure from the ignorability assumption: (i) one unmeasured confounder independent of measured confounders,  (ii) one unmeasured confounder with interactions with two measured confounders, and (iii) two unmeasured confounders with interactions between each other and with two measured confounders. In the SEER-Medicare data, at least two confounders (lung function and physical activity level) \citep{saito2017impact,bille2021preoperative} are not collected, and there may exist interactions among measured and unmeasured confounders. This situation is represented by level (iii).  The third factor is covariate overlap, the degree of which we varied  in three scenarios: strong overlap, moderate overlap and weak overlap. The covariate overlap in the SEER-Medicare data is reasonably strong (Web Figure 5). For the sake of simulation efficiency, we laid this factor into the sample size $N=10,000$ with ratio = 1:10:9 -- the scenario most representative of the SEER-Medicare data.   In this set of simulations, we considered  nonlinear treatment assignment mechanism,  nonparallel and nonlinear response surfaces and partial alignment between predictors of the treatment assignment mechanism and predictors of the response surfaces, which are realistic scenarios representative of the real-world data. 

We now describe the data generating processes and parameter constellations, which are based on the observed information in the SEER-Medicare data. In addition, the parameters in the treatment assignment and outcome generating mechanisms are chosen in parallel to the design factors and to generate similar outcome event probabilities as in the SEER-Medicare data. We simulated 15 covariates $X=(X_1, \ldots, X_{15})$, among which $X_6, X_7, X_8, X_9, \text{ and } X_{10}$ were discrete and others were continuous. The distributions of the covariates are summarized in Web Table 1. The treatment assignment mechanism follows a multinomial logistic regression model, 
$$A \cond X^A \sim \text{Multinomial} \lp N, \pi_1(X^A), \pi_2(X^A), \pi_3(X^A) \rp,$$
where $\pi_j (X^A)= e^{l_j(X^A)}/\sum_{j=1}^3 e^{l_j(X^A)}$, $\log \lp l_j(X^A) \rp = \alpha_j + \gamma \lp X^A \xi^L_j + Q^{A} \xi_j^{NL} \rp$ $\forall j \in \{1,2,3\}$,  $X^A$ is a subspace of $X$ predictive of the treatment, $Q^A$ denotes the nonlinear transformations and higher-order terms of the predictors $X^A$, and $\xi^{L}_j$  and $\xi^{NL}_j$ are respectively vectors of coefficients for $X^A$ and $Q^A$. The intercepts $(\alpha_1, \alpha_2, \alpha_3)$ were used to control the ratio of units across three treatment groups and $\gamma$ was varied to create different levels of covariate overlap. Using data simulated under the 1:10:9 ratio,  we varied $\gamma$ to reflect increasingly sparse covariate overlap among treatment groups.  Web Figure 4 shows the distributions of the true generalized propensity scores across treatment groups.

Three sets of nonparallel response surfaces were generated as follows: 
\begin{eqnarray*}
\mathbb{P}\lp Y(j) =1 \cond X^Y \rp &=& \text{logit}^{-1} \lp \tau_j + X^Y \eta^L_j + Q^Y \eta^{NL}_j \rp \forall j \in \{1, 2, 3\}, 
\end{eqnarray*}
where $X^Y$ is a subspace of $X$ predictive of the outcome (partially aligned with $X^A$), $Q^Y$ is the nonlinear transformations and higher-order terms of  $X^Y$, and $(\eta^{L}_j, \eta^{NL}_j)$ are vectors of coefficients for $(X^Y, Q^Y)$.  The parameters $(\tau_j, \eta_j^L, \eta_j^{NL})$ were set to differ across $j$'s to induce treatment heterogeneity.  We chose the parameter values so that the outcome event rate in each treatment group was similar as the postoperative complication rate observed in the SEER-Medicare data.  The observed outcome event probability is 0.38 in treatment group 1, 0.34 in treatment group 2 and 0.51 in treatment group 3.  The true pairwise $\text{CATE}_{1,2}=0.05$, $\text{CATE}_{1,3}=-0.11$ and $\text{CATE}_{2,3}=-0.16$. 

To create unmeasured confounding, for case (i) we chose $X_4$, independent of all  measured covariates, as an unmeasured confounder. For case (ii) we chose $X_{13}$ as an unmeasured confounder, which had interactions with both $X_{11}$ and $X_{12}$, and for case (iii) we denied access to  $X_{14}$ and $X_{15}$, which had interactions with each other and with both $X_{11}$ and $X_{12}$.

There are a total of 8 configurations for this simulation: \big($N=1500$ with ratio =1:1:1,  $N=10000$ with  ratio = 1:10:9\big) $\times$ \big(UMC (i), UMC(ii), UMC(iii)\big)  under strong covariate overlap $+$ \big($N=10000$ with  ratio = 1:10:9\big) $\times$ \big(weak overlap, moderate overlap\big).

\begin{figure}[H] 
\includegraphics[width=1\textwidth]{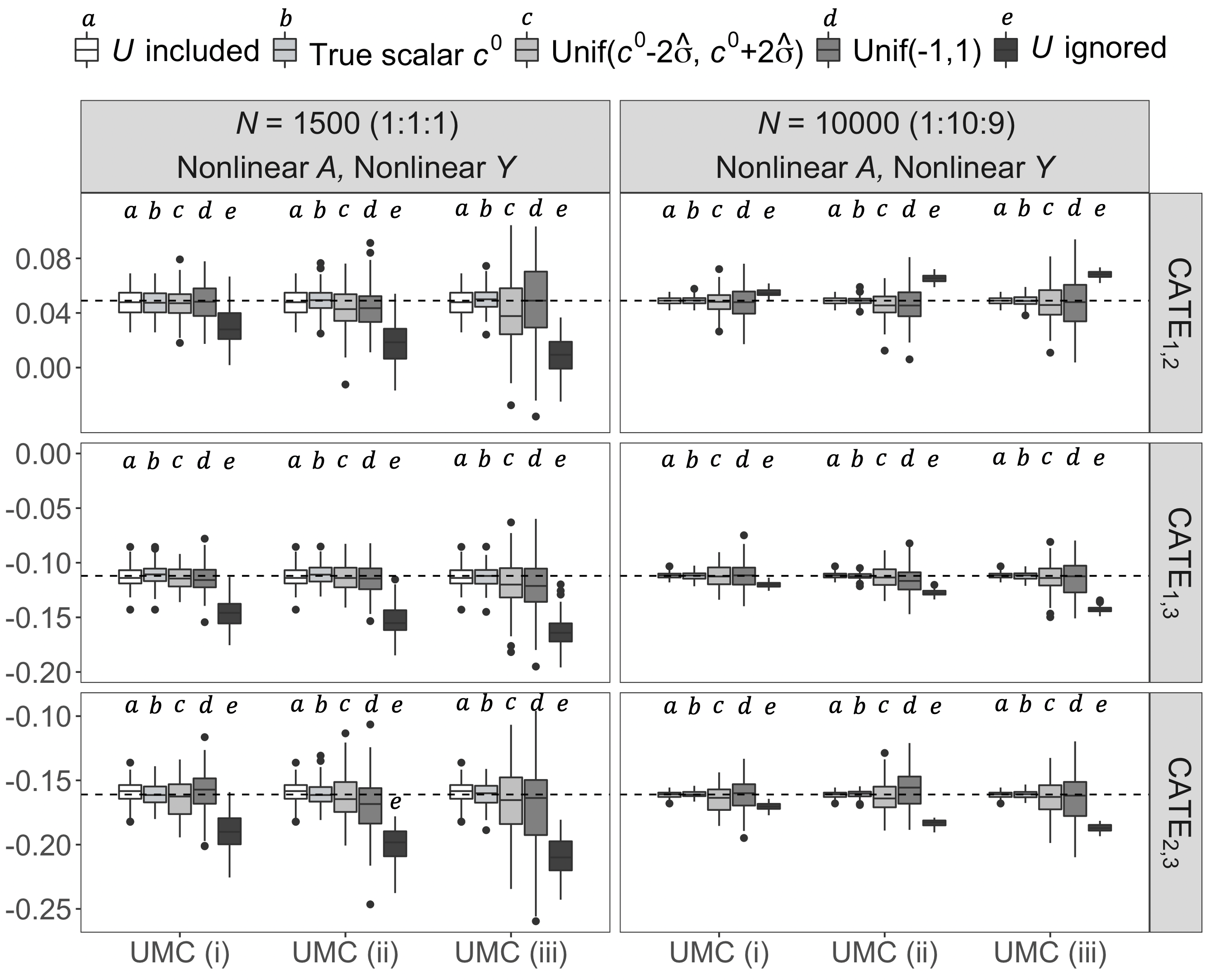}
\caption{Estimates of three pairwise causal  effects $\text{CATE}_{1,2}$, $\text{CATE}_{1,3}$ and $\text{CATE}_{2,3}$ among 1000 replications, for two sample sizes, $N=1500$ with ratio = 1:1:1 and $N=10,000$ with ratio = 1:10:9, and three complexity levels of unmeasured confounding, UMC(i), UMC(ii) and UMC(iii), as described in Section~\ref{sec:complex-sim-design}. For sensitivity analysis, three strategies were used to specify the prior distributions for the confounding functions: 1) true scalar parameters $c^0$, 2) $\mathcal{U}(c^0-2\hat{\sigma}, c^0+2\hat{\sigma})$ bounded within $[-1,1]$, and 3) $\mathcal{U}(-1,1)$.  The CATE results that could be achieved if $U$ were actually observed and the naive CATE estimators ignoring $U$  are also presented. Red dashed lines correspond to the true CATE.} 
\label{fig: complex-sim}
\end{figure}

\subsubsection{Simulation results: Overall assessment}
Figure~\ref{fig: complex-sim} shows the estimates of three pairwise causal treatment effects among 1000 replications for six scenarios generated under a combination of  sample sizes and complexity levels of unmeasured confounding. The covariate overlap is relatively strong in all six data configurations. Using the true scalar parameters for the $c(\cdot)$ functions, our confounding function adjusted estimates of the casual effects are close to the estimates were the treatment assignment ignorable, even under the most complex unmeasured confounding (iii), for all pairwise treatment effects and both sample sizes. 
The effect estimates are more accurate with a larger sample size.  When the uncertainty about the sensitivity parameters increases, the range of adjusted causal effects widens. The smaller sample size is more susceptible to this impact. For example, the whiskers of boxplots for the $\mathcal{U}(-1,1)$ and $\mathcal{U}(c^0-2\hat{\sigma}, c^0+2\hat{\sigma})$ specifications extend into the opposite direction for the $\text{CATE}_{1,2}$ effect in the scenario with the most complex level of unmeasured confounding. The wide breadth of the adjusted estimates of the causal treatment effects suggest that calibrating unmeasured confounding for reliable input by building in subject-matter expertise could be valuable, especially for smaller sample size. The poor performance of the naive estimators ignoring the violation of the ignorability assumption is evident in all scenarios. Web Table 3 summarizes the average absolute bias and RMSE in each of the estimates of treatment effects for each sensitivity analysis strategy and  data configuration. 

\begin{figure}[H] 
\includegraphics[width=1\textwidth]{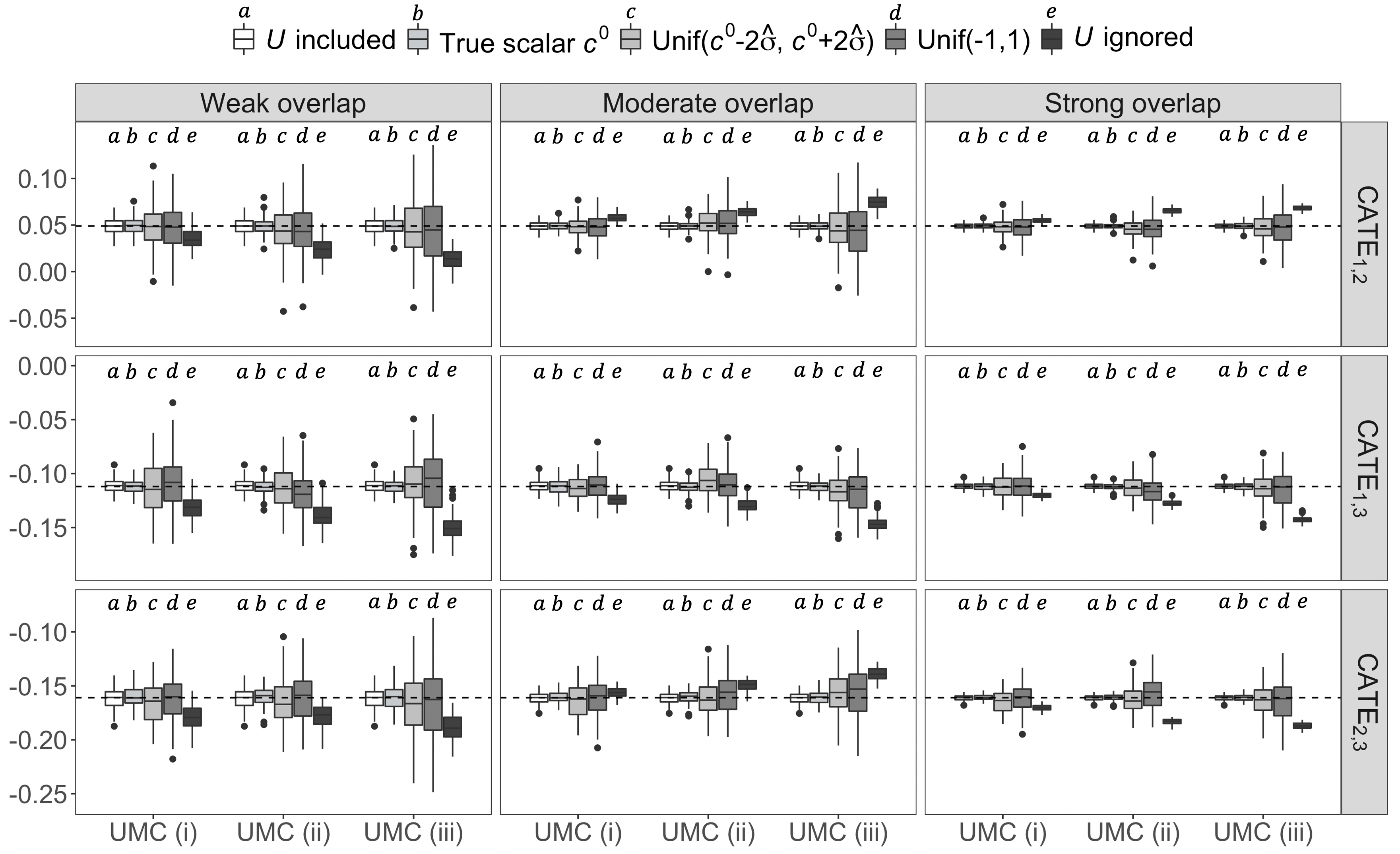}
\caption{Estimates of three pairwise causal treatment effects $\text{CATE}_{1,2}$, $\text{CATE}_{1,3}$ and $\text{CATE}_{2,3}$ among 1000 replications, for the sample size $N=10,000$ with ratio of units = 1:10:9, three complexity levels of unmeasured confounding, UMC(i), UMC(ii) and UMC(iii), and three levels of covariate overlap, strong, moderate and weak,  as described in Section~\ref{sec:complex-sim-design}. For sensitivity analysis, three strategies were used to specify the prior distributions for the confounding functions: 1) true scalar parameters $c^0$, 2) $\mathcal{U}(c^0-2\hat{\sigma}, c^0+2\hat{\sigma})$ bounded within $[-1,1]$, and 3) $\mathcal{U}(-1,1)$.  The CATE results that could be achieved if $U$ were actually observed and the naive CATE estimators ignoring $U$  are also presented. Red dashed lines mark the true CATE.} 
\label{fig: sim-result-overlap}
\end{figure}

\subsubsection{Simulation results: Covariate overlap} 
 Figure~\ref{fig: sim-result-overlap} demonstrates whether and how covariate overlap impacts the operating characteristics of our sensitivity analysis approach under different complexity levels of unmeasured confounding. In all levels of overlap, our sensitivity analysis estimators, assuming the true values of $c$ functions, are close to the causal effect estimates if the igorability assumption holds. When the covariate overlap becomes increasingly sparse, the sensitivity analysis estimators assuming large uncertainties about the sensitivity parameters produce highly variable adjusted causal effect estimates, with the variability in proportion to the complexity level of unmeasured confounding. In the meanwhile, ignoring $U$ could lead to extremely deviated treatment effect estimates. 

\subsubsection{Simulation results: Coverage probability} 
In the situation where there is at least moderate covariate overlap and the sample size is relatively large $N=10,000$, our sensitivity analysis estimators yield close-to-nominal coverage probabilities (Figure~\ref{fig: sim-complex-CP}). When there is a substantial lack of overlap, a specification with a large uncertainty for the sensitivity parameters tends to yield lower coverage probabilities. However, even under the most complex UMC (iii) scenario, the coverage probabilities are around 0.88. On the contrary, naive estimators ignoring unmeasured confounding  would give failed inference with essentially zero coverage probabilities across all configurations. Similar observations are made about the coverage probabilities for the smaller sample size $N=1500$, shown in Web Figure 6. Due to the smaller sample size, the coverage probabilities of all estimators are slightly lower than those for the larger sample size $N=10,000$ across all scenarios.

\begin{figure}[H] 
\includegraphics[width=0.9\textwidth]{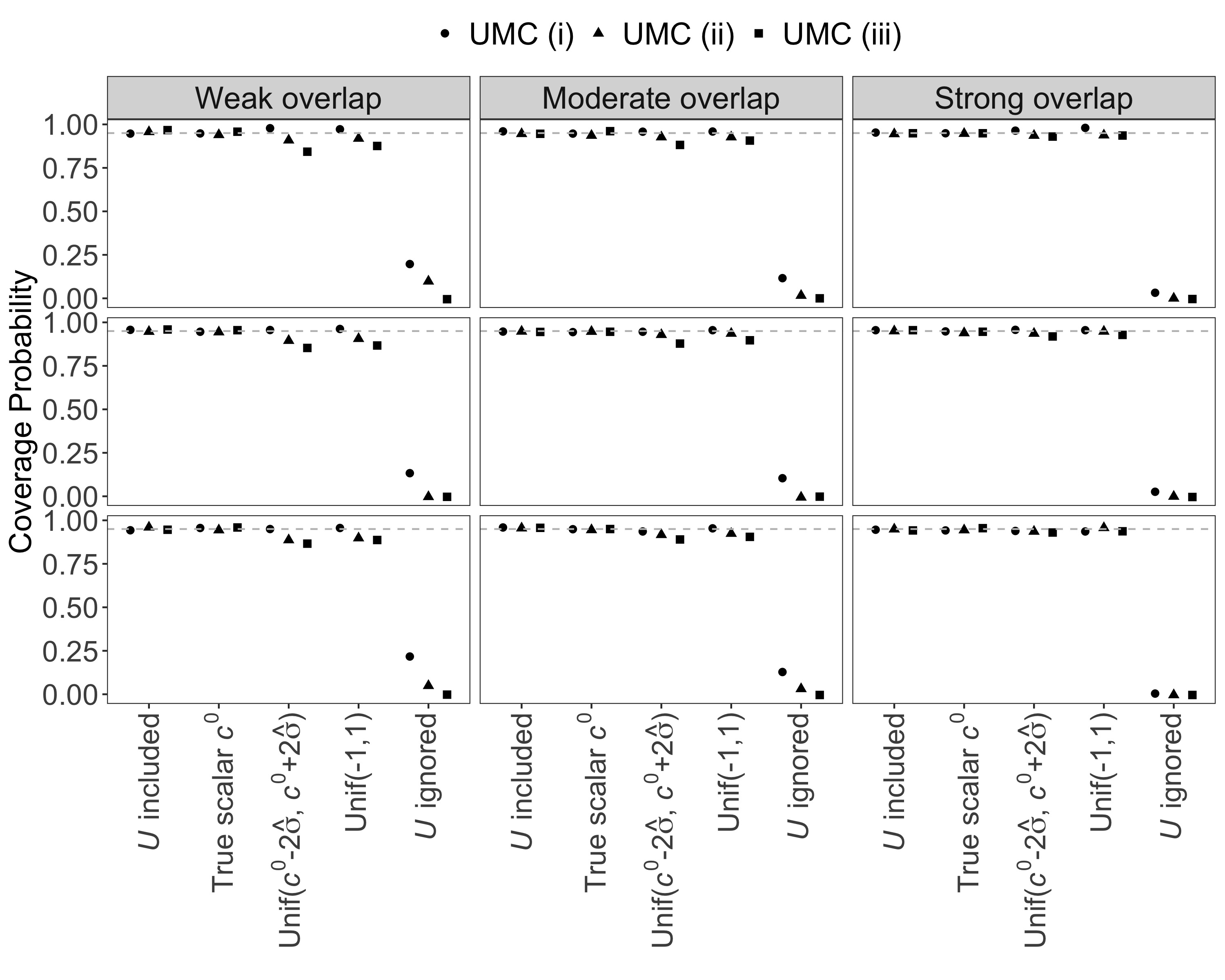}
\caption{The coverage probability of the estimates of three pairwise causal  effects $\text{CATE}_{1,2}$, $\text{CATE}_{1,3}$ and $\text{CATE}_{2,3}$ among 1000 replications, for the sample size $N=10,000$ with ratio of units = 1:10:9, three complexity levels of unmeasured confounding, UMC(i), UMC(ii) and UMC(iii), and three levels of covariate overlap, strong, moderate and weak,  as described in Section~\ref{sec:complex-sim-design}. For sensitivity analysis, three strategies were used to specify the prior distributions for the confounding functions: 1) true scalar parameters $c^0$, 2) $\mathcal{U}(c^0-2\hat{\sigma}, c^0+2\hat{\sigma})$ bounded within $[-1,1]$, and 3) $\mathcal{U}(-1,1)$. Coverage probabilities of the CATE results that could be achieved if $U$ were actually observed and the naive CATE estimators ignoring $U$  are also presented. Gray dashed lines mark the nominal 95\% coverage probability.} 
\label{fig: sim-complex-CP}
\end{figure}

\section{Sensitivity analysis of causal effects of NSCLC surgical approaches}\label{sec:application} 
Our SEER-Medicare  data set includes 11,980 patients older than 65 years, who were diagnosed with stage I--IIIA NSCLC between 2008 and 2013 and underwent surgical resection via RAS ($n=396$),  OT ($n=5002$) or VATS ($n=6582$). The data set contains individual-level information at baseline on the following variables: age, sex,  race, ethnicity, marital status, income level, Charlson comorbidity score, cancer stage, tumor size, tumor site, cancer histology, and whether positron emission tomography, chest computed tomography, or mediastinoscopy was performed.  The postoperative complication outcomes  are (i) the presence of postoperative respiratory complication within 30 days of surgery or during the hospitalization in which the primary surgical procedure was performed; (ii) prolonged LOS; (iii) ICU stay following surgery; and (iv) readmission within 30 days of surgery. The observed respiratory complication rate is 30.1\% in the RAS group, 33.3\% in the OT group and 33.6\% in the VATS group. The proportion of prolonged LOS is observed to be 5.3\% in RAS, 5.5\% in OT and 10.4\% in VATS. For ICU stay, the observed event rate is 60.2\% in RAS, 59.1\% in OT and 75.3\% in VATS. And 8.8\% of patients were readmitted within 30 days of surgery via RAS, 8.0\% via OT and 9.8\% via VATS. Detailed data description can be found in \cite{hu2020estimation}. The covariate overlap is relatively strong in the data (Web Figure 5).

Prior causal inference research using the SEER-Medicare NSCLC data found that VATS led to a significantly smaller respiratory complication rate, lower chance of prolonged LOS and  lower chance of ICU stay than OT, but there was no statistically significant difference between RAS and VATS or between RAS and OT. No surgical approach was significantly better than the others regarding 30-day readmission rate \citep{hu2020estimation}.   Given the possible unmeasured confounders like preoperative lung function and physical activity level \citep{sihoe2020video,bille2021preoperative}, we apply the proposed sensitivity analysis approach to evaluate the sensitivity of the estimated causal effects to the potential magnitude of departure from ignorable treatment assignment. 

We leverage the subject-area literature to specify the prior distributions for the confounding functions. We first consider the signs of the confounding functions.  A large volume of evolving clinical evidence
has confirmed that VATS lung cancer resection offered
proven safety and feasibility and has firmly established VATS
as the surgical approach of choice for early-stage lung
cancer today  \citep{sihoe2020video}. On the other hand, it is recommended that appropriate patient selection is essential for success
of RAS. Severe cardiac and pulmonary disease
should be considered contraindications to robotic surgery
as these patients will not tolerate one-lung
ventilation or changes in venous return \citep{VATS5336}. It has been shown that OT is less safe than VATS \citep{howington2013treatment} and a patient who is high risk for OT is also high risk for RAS \citep{VATS5336}. Based on these pieces of evidence, we posit that the unmeasured factors guiding clinicians to choose an approriate surgical approach tend to lead them to recommend RAS or OT over VATS to healthier patients, that is, $c(1,3,x)<0$, $c(3,1,x)>0$, $c(2,3,x)<0$ and $c(2,3,x)>0$ (Table~\ref{tab:cfun-inter}). There is no sufficient domain knowledge to inform a particular direction for $c(1,2,x)$ or $c(2,1,x)$. We next consider the bounds of the confounding functions. \cite{howington2013treatment} show that preoperative pulmonary lung function has a large effect (Cohen's $d >1.2$) and physical activity level has a small effect (Cohen's $d =0.2$) on postoperative complications. Roughly translating  Cohen's $d$ to generalized eta squared $\eta^2_G$ \citep{lakens2013calculating}, which measures the proportion of the total variation in the outcome explained by a given covariate, the unmeasured confounders approximately account for 30\% of the total variance in postoperative complication outcomes. The $R^2$ of a regression model fitted to the SEER-Medicare NSCLC data suggest that 40\% of the variation in overall postoperative complications are unexplained by measured covariates. Based on these grounds, we establish the bounds for the $c(\cdot)$ functions by assuming that the unmeasured confounding would account for less than $h=0.8$ units of the remaining standard deviation unexplained by measured $X$.

 We used  six strategies to specify the prior distributions on the confounding functions to conduct a comprehensive sensitivity analysis.  Table~\ref{tab:SA-resp} displays the sensitivity analysis results for postoperative respiratory  complication. Prior domain knowledge led us to believe that clinicians tend to prescribe RAS or OT over VATS to healthier patients, and that the unmeasured confounders account for no more than 0.8 units of the remaining standard deviation unexplained by measured confounders, or approximately 0.4 for postoperative respiratory complications. We first look at the effect between RAS and VATS. Under specification (i)-(iii),  the comparative benefit in terms of postoperative respiratory complication of VATS versus RAS becomes statistically significant, as opposed to the conclusion that would have been drawn assuming ignorable treatment assignment (specification [vii]). Under the assumption that unmeasured factors tend to lead clinicians to systematically prescribe RAS to healthier patients relative to VATS, $c(1,3,x)  \sim \mathcal{U} (-0.4,0) \text{ and } c(3,1,x) \sim \mathcal{U}(0, 0.4)$, and that the observed treatment allocation between RAS or OT is beneficial relative to the alternative which reverses treatment assignment for everyone (specification [iii]) regardless of clinician's preference between RAS and OT (specification [i] and [ii]),
the naive estimator ignoring unmeasured confounding (specification [vii]) is biased downwards. The ``bias correction'' in our sensitivity analysis approach raises the potential risk of postoperative respiratory complication associated with RAS, leading to statistically significant benefit of VATS. Turning to the comparative effect of RAS and OT, assuming clinicians tend to systematically prescribe RAS to the healthiest patients (specification [i])  would lead to an altered conclusion that OT is more beneficial than RAS. If we presume that clinicians would prefer to prescribe OT to healthier patients relative to RAS (specification [ii]), then the adjusted causal effect estimates suggest that RAS is more beneficial than OT. Finally, when comparing OT with VATS, assuming the effect of unmeasured confounding between RAS and OT can be either positive or negative (specification [v]), the significant treatment benefit associated with VATS relative to OT under the ignorability assumption is negated. The larger uncertainty about the confounding functions are reflected in the wider uncertainty intervals for the adjusted causal effect estimates. 

In situations where researchers may be interested in assessing the sensitivity of the CATT results to various degrees of unmeasured confounding, we can follow the same steps 1-3 in Algorithm~\ref{alg:MCSA}, and modify step 4 to obtain the adjusted CATT results by averaging the differenced potential outcomes among those in the reference group. Web Table 7 shows sensitivity analysis results for the CATT effects among those who were operated with RAS. The same specifications of the confounding functions were used as for the CATE effects in Table~\ref{tab:SA-resp}. Though the CATT results are less sensitive than the CATE results to violations of the ignorability assumption, demonstrated by the smaller magnitude of changes in the effect estimates, the causal conclusion about the comparative effect of RAS versus VATS changed (from no significant difference to VATS being significantly better) under specifications (i) and (iii). 

\begin{table}[H]
\centering
\caption{Sensitivity analysis for causal inferences about average treatment effects of three surgical approaches on postoperative respiratory complications based on the risk difference, using the SEER-Medicare lung cancer data. Three surgical approaches are $A=1$: robotic-assisted surgery (RAS), $A=2$: open thoracotomy (OT), $A=3$: video-assisted thoracic surgery (VATS).  The adjusted effect estimates and 95\% uncertainty intervals are displayed. Interval estimates are based on pooled posterior samples across model fits arising from $30 \times 30$ data sets.  The remaining standard deviation in the outcome not explained by the measured covariates is $\hat{\sigma}$ = 0.46 . We assume  $c(1,3,x)  \sim \mathcal{U} (-0.4,0), \; c(3,1,x) \sim \mathcal{U}(0, 0.4), \; c(2,3,x) \sim \mathcal{U} (-0.4,0), \; c(3,2,x) \sim \mathcal{U}(0, 0.4)$ for specification (i)-(v). }
\label{tab:SA-resp}
\begin{tabular}{cp{0.42\textwidth}ccc} 
\toprule
&Prior distributions on $c(\cdot)$ functions  & RAS vs. OT & RAS vs. VATS & OT vs. VATS \\\hline
(i)& $c(1,2,x) \sim \mathcal{U} (-0.4,0),  c(2,1,x) \sim \mathcal{U}(0, 0.4)$ &  $.03 (.01,.05)$ & $.05(.03,.07)$ & $.06 (.04,.08)$\\
(ii)&$c(1,2,x) \sim \mathcal{U}(0, 0.4),  c(2,1,x) \sim \mathcal{U} (-0.4,0)$  &  $-.03 (-.05,-.01)$ & $.02(.00,.04)$ & $.02 (.00,.04)$\\
(iii)&$c(1,2,x), c(2,1,x) \sim \mathcal{U} (-0.4,0) $ & $.04(.02,.06)$ & $.04(.02,.06)$ & $.05 (.03,.07)$\\
(iv)&$c(1,2,x), c(2,1,x) \sim \mathcal{U}(0, 0.4)$ & $.00(-.02,.02)$ & $.01(-.01,.03)$ & $.03 (.01,.05)$\\
(v) &$c(1,2,x), c(2,1,x) \sim \mathcal{U}(-0.4, 0.4)$ & $.01(-.02,.04)$ & $.02(-.01,.05)$ & $.03 (-.00,.06)$\\
(vi)&all $c(\cdot) \sim \mathcal{U}(-1, 1)$ & $.01(-.06,.08)$ & $.07(-.00,.14)$ & $.06(-.01,.13)$ \\
(vii)&all $c(\cdot) =0$ & $-.01(-.02,.00)$ & $.01(-.00,.02)$ & $.02(.01,.03)$\\
\bottomrule
\end{tabular}
\end{table}

Web Table 4-6 summarize sensitivity analysis results for prolonged LOS, ICU stay and 30-day readmission under specifications (i)-(vi) of the confounding functions. Plausible assumptions about the sensitivity parameters altered causal conclusions about the treatment benefit of VATS versus RAS (from nonsignificant to significant) in terms of prolonged LOS; and about treatment benefits of VATS versus OT or RAS (from nonsignificant to significant) with respect to 30-day readmission.  Based on ICU stay, the estimates of causal effects are insensitive to various magnitudes of departure from the ignorability assumption.

To provide more granularity to the sensitivity analysis, Figure~\ref{fig: heatmap} disaggregates the sensitivity analysis results for the causal effect of RAS versus VATS by  combinations of four $c(\cdot)$ functions involved in the pairwise treatment effect and displays them in a contour plot \citep{brumback2004sensitivity, li2011propensity,  kasza2017assessing}. As discussed above, we leverage domain knowledge and assume healthier patients receive OT relative to VATS and there is a lack of guidance for specifying the direction of unmeasured confounding between RAS and OT, that is, $c(3,2) \sim \mathcal{U} (0, 0.4)$ and $c(1,2) \sim \mathcal{U} (-0.4, 0.4)$. Subject-area literature provides evidence to support the assumption that relative to VATS,  patients prescribed to RAS tend to be healthier. We examine how the causal conclusion about the effect of RAS versus VATS would change under different pairs of values of $c(1,3,x)$ and $c(3,1,x)$ spaced on a grid $(-0.4,0) \times (0, 0.4)$ and the specifications of $c(1,2,x)$ and $c(3,2,x)$. When there is no unmeasured confounding between RAS and VATS, i.e., $c(1,3,x)= 0$ and $c(3,1,x)=0$, the estimate estimate of $CATE_{1,3}$ suggests that the probability of having postoperative respiratory complications under RAS is .006 lower than that under VATS. When the treatment assignment between RAS and VATS is not ignorable, as the magnitudes of $c(1,3,x)$ and $c(3,1,x)$ increase, i.e. patients prescribed to RAS are healthier, the adjusted effect estimates indicate larger treatment benefit with VATS.
Similar patterns are observed in Web Figure 7 for the other two pairwise treatment effects comparing RAS with OT and OT with VATS, suggesting the sensitivity to potential unmeasured confounding of the causal effect estimates on postoperative respiratory complications. 

 \begin{figure}[H] 
\includegraphics[width=0.7\textwidth]{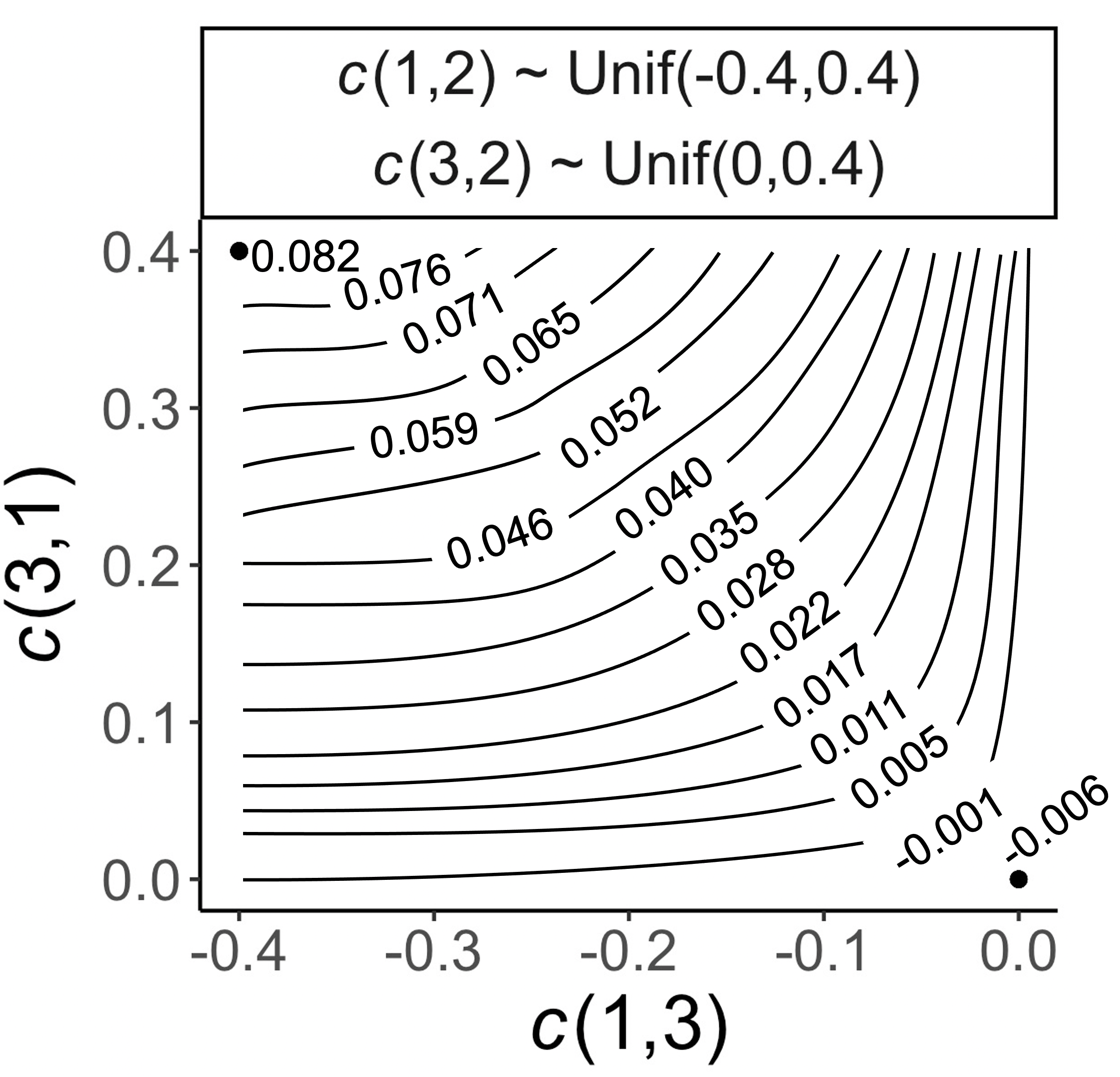}
\caption{Contour plot of the confounding function adjusted treatment effect estimate for RAS versus VATS, $CATE_{1,3}$. Four confounding functions are involved, $c(1,3,x)$, $c(3,1,x)$, $c(1,2,x)$, $c(3,2,x)$. The  black lines report the adjusted causal effect estimates corresponding to pairs of values for  $c(1,3,x)$ and $c(3,1,x)$ spaced on a grid $(-0.4, 0) \times (0,0.4)$ incremented by 0.02, and under the prior distributions $c(3,2,x) \sim \mathcal{U}(0,0.4) $ and $c(1,2,x) \sim \mathcal{U}(-0.4,0.4)$.} 
\label{fig: heatmap}
\end{figure}

\section{Discussion} \label{sec:discussion}
We propose a Monte Carlo sensitivity analysis approach for causal inference with multiple treatments and a binary outcome. In such settings, few off-the-shelf sensitivity analysis methods are available. The proposed method is readily available in the $\R$ package  $\textsf{SAMTx}$. Our approach has two main advantages. First, nested MI is used to draw values of the confounding functions from their user-specified prior distributions and draw values of the generalized propensity scores (a component of the bias formula) from their predictive posterior distributions, thereby  the uncertainty about the values of the sensitivity parameters and the statistical uncertainty due to sampling  are seamlessly amalgamated in the analysis. Our sensitivity analysis estimators are based on modularized Bayesian inference,  without contaminating the generalized propensity score model and distorting inferences about the causal effect \citep{zigler2013model}. Second, confounding functions have previously been paired with the inverse probability weighting for a sensitivity analysis \citep{robins1999association, brumback2004sensitivity, li2011propensity}, inheriting the issues of extreme weights and the potential inefficiency of the weighting based estimator. Our approach estimates causal effect estimates using outcome modeling via machine learning, substantially increasing modeling flexibility and efficiency. 

 We use simulations to gain insight into the operating characteristics of our proposed sensitivity analysis approach in the  complex multiple treatment setting. 
There are five key findings. First, formulating the sensitivity parameters in respect of the remaining standard deviation is a practical and useful strategy. 
Second, in a sensitivity analysis for unmeasured confounding with multiple treatments, considerations of all pairs of treatments are needed (Figure~\ref{fig: binary-sim}, I: 3rd $A$ ignored). Third, with a binary outcome, our sensitivity parameters are naturally bounded on the interval $[-1,1]$, offering a workable  noninformative specification of the sensitivity parameters for the situation where the domain knowledge about unmeasured confounding is absent. Simulations demonstrate that, in the presence of unmeasured confounding,  our sensitivity analysis estimators achieve substantially larger bias reduction and better coverage probability than the naive estimators ignoring unmeasured confounding.  Specifications with large uncertainties about the values of the sensitivity parameters, which cover the extreme scenarios of particularly strong degrees of unmeasured confounding, could yield highly variable adjusted causal conclusions. Subject-matter expertise is needed to judge whether these extreme scenarios are plausible.
Fourth, formulating the confounding function as a scalar parameter is a useful simplifying strategy for specifying the prior distributions for the confounding functions when there are many measured confounders. Our sensitivity analysis estimators in the contextualized simulations are similar to the causal effect estimates as if we had access to the unmeasured confounders, provided that the confounding functions are known.
Fifth, a uniquely new investigation into the statistical properties of our sensitivity analysis estimators demonstrate that when there is at least moderate covariate overlap, our sensitivity analysis estimators provide close-to-nominal coverage probability even under complex mechanisms of unmeasured confounding. As the goal of a sensitivity analysis is to adjust the causal conclusions about treatment effects, the underlying assumption of covariate overlap (A2) is still required for drawing valid causal inferences. One may need to apply techniques for identifying a common support region \citep{hu2020estimation} for retaining inferential units and for avoiding extrapolating over areas of the covariate space where common support does not exist before carrying out a sensitivity analysis.

The comprehensive sensitivity analysis of the SEER-Medicare  data elucidates the comparative causal effects of three popular surgical approaches for treating early stage NSCLC tumors based on four postoperative outcomes. Various assumptions about the potential magnitude and direction of departure from the ignorability assumption were used to evaluate the sensitivity of the estimated causal effects to unmeasured confounding. While causal conclusions can be sensitive for respiratory complication and 30-day readmission, they were relatively robust against a variety of violations of the ignorability assumption with respect to prolonged LOS and ICU stay. 

 Our approach can be extended in several directions. First, 
 though the constant parameterization of our sensitivity parameters works well in a wide variety of settings, developing a low-dimensional confounding function containing summarizing or key information of the measured confounders may increase the possibility of uncovering the true form of the confounding function \citep{li2011propensity}. Second, developing a fully Bayesian sensitivity analysis approach could be a worthwhile  contribution. While our near-Bayesian approach is able to incorporate the uncertainty about sensitivity parameters into the  analysis, the application of Bayes theorem in a non-identifiable model can sometimes rule out certain patterns of unmeasured confounding that are incompatible with data, consequentially leading to more consistent 95\% frequentist coverage probability.  \citep{mccandless2017comparison}.  Finally, future work on enhancing BART for causal inference with multiple treatments and binary outcomes could potentially improve the performance of our sensitivity analysis approach.  One promising avenue is allowing separate priors inducing different levels of regularization, and consequentially reducing regularization-induced bias in treatment effect estimates \citep{hahn2020bayesian}.

\section{Software}
Software in the form of $\R$ package $\code{SAMTx}$ to implement our methods and replicate the simulations is available at \url{https://cran.r-project.org/web/packages/SAMTx/index.html}. 

\begin{appendix}
\section*{Proofs of theorems}\label{appn} 
 \begin{proof}[Proof of Theorem~\ref{thm:biasform}]
 Under nonignorable treatment assignment, ignoring unmeasured confounding will lead to the following bias in the estimate of the causal effect $CATE_{a_j,a_k}$, 
$$\text{Bias} = \E[Y|A=a_j, X=x] -\E[Y|A=a_k, X=x]-\E\lsq Y(a_j)-Y(a_k)\cond X=x\rsq.$$ To simplify notation, we will use $\E\lsq \cdot \cond a, x\rsq $ to denote $\E\lsq \cdot \cond A=a, X=x\rsq $. 
Applying the law of total expectation to $\E\lsq Y(a_j \cond X=x)\rsq$, we have 
\begin{eqnarray*}
\E\lsq Y(a_j) \cond X=x\rsq &=& \sum_{m=0}^J p_{m}\E\lsq Y(a_j)\cond a_m, x \rsq. 
\end{eqnarray*}
Similary, 
\begin{eqnarray*}
\E\lsq Y(a_k) \cond X=x\rsq &=& \sum_{m=0}^J p_{m}\E\lsq Y(a_k)\cond a_m, x \rsq. 
\end{eqnarray*}
Therefore, 
\begin{eqnarray} \nonumber
\E \lsq Y(a_j) - Y(a_k) \cond x  \rsq & = &p_0 \E \lsq Y(a_j)-Y(a_k) \cond a_0,x \rsq + \ldots + p_J \E \lsq Y(a_j)-Y(a_k) \cond a_J,x \rsq  \\ \label{eq:twoEs}
     &&+ p_j \E \lsq Y(a_j)-Y(a_k) \cond a_j,x \rsq +  p_k \E \lsq Y(a_j)-Y(a_k) \cond a_k,x \rsq.  
\end{eqnarray}
We will repeatedly use $\E \lsq Y(a_m) \cond a_m, x \rsq = \E \lsq Y \cond a_m, x \rsq,  \forall m \in \{1, \ldots, J\}$.  Rewriting the last two items of the RHS of the equation~\eqref{eq:twoEs}  yields 
\begin{eqnarray}\nonumber
&&p_j \E \lsq Y(a_j)-Y(a_k) \cond a_j,x \rsq +  p_k \E \lsq Y(a_j)-Y(a_k) \cond a_k,x \rsq\\\nonumber
&=&p_j \E \lsq Y(a_j) \cond a_j, x \rsq - p_j \E \lsq Y(a_k) \cond a_j, x \rsq + p_k \E \lsq Y(a_j) \cond a_k,x \rsq - p_k \E \lsq Y(a_k) \cond a_k,x \rsq \\\nonumber
&=& p_j \E \lsq Y \cond a_j, x \rsq - p_k \E \lsq Y \cond a_k, x \rsq + p_k \E \lsq Y (a_j)\cond a_k, x \rsq - p_j \E \lsq Y (a_k)\cond a_j, x \rsq  \\\nonumber
&=& p_j \lbc \E \lsq Y \cond a_j, x \rsq -  \E \lsq  Y \cond a_k, x \rsq \rbc + p_j  \E \lsq  Y \cond a_k, x \rsq - p_k \E \lsq Y \cond a_k, x \rsq \\\nonumber
&& - p_j \E \lsq Y (a_k)\cond a_j, x \rsq + p_k \E \lsq Y (a_j)\cond a_k, x \rsq\\\nonumber
&=&  p_j \lbc \E \lsq Y \cond a_j, x \rsq -  \E \lsq  Y \cond a_k, x \rsq \rbc + p_j \lbc  \E \lsq Y(a_k) \cond a_k,x \rsq - \E \lsq Y(a_k) \cond a_j,x \rsq  \rbc \\\nonumber
&&+ p_k \lbc  \E \lsq  Y(a_j) - Y(a_k) \cond a_k, x\rsq \rbc \\ \label{eq: eqpk}
&=&p_j \lbc \E \lsq Y \cond a_j, x \rsq -  \E \lsq  Y \cond a_k, x \rsq \rbc +p_j c(a_k, a_j, x) +  p_k \lbc  \E \lsq  Y(a_j) - Y(a_k) \cond a_k, x\rsq \rbc. 
\end{eqnarray}

Let $\tilde{p} = 1- p_j - p_k$.  By rewriting $p_k \lbc  \E \lsq  Y(a_j) - Y(a_k) \cond a_k, x\rsq \rbc$  in equation~\eqref{eq: eqpk}, we have 
\begin{eqnarray*}
&&p_k \lbc  \E \lsq  Y(a_j) - Y(a_k) \cond a_k, x\rsq \rbc \\
&=& (1-p_j-\tilde{p}) \lbc \E \lsq Y \cond a_j ,x  \rsq - \E \lsq Y \cond a_k, x \rsq  + \E \lsq  Y(a_j) \cond a_k, x \rsq  - \E \lsq  Y(a_j) \cond a_j, x\rsq \rbc\\
&=& (1-p_j)  \lbc \E \lsq Y \cond a_j ,x  \rsq - \E \lsq Y \cond a_k, x \rsq \rbc  - (1-p_j - \tilde{p}) c(a_j, a_k, x)  - \tilde{p} \lbc \E \lsq Y \cond a_j, x \rsq - \E \lsq Y \cond a_k, x \rsq  \rbc \\
&=& (1-p_j)  \lbc \E \lsq Y \cond a_j ,x  \rsq - \E \lsq Y \cond a_k, x \rsq \rbc - p_k c(a_j, a_k,x) \\
&&- (1-p_j-p_k) \lbc \E \lsq Y (a_j) \cond a_j, x \rsq - \E \lsq Y(a_k) \cond a_k, x \rsq  \rbc.
\end{eqnarray*}
Taken together, 
\begin{eqnarray*}
\text{Bias} &= &\sum_{l: l \in \mathcal{A} \setminus\{a_j, a_k\} } -p_l \E \lsq  Y(a_j) - Y(a_k) \cond a_l, x\rsq - p_j c(a_k, a_j, x) + p_k c(a_j, a_k, x) \\
&& + \sum_{l: l \in \mathcal{A} \setminus\{a_j, a_k\} } p_l\lbc  \E\lsq Y(a_j) \cond a_j, x \rsq  - \E\lsq Y(a_k) \cond a_k, x \rsq \rbc\\
&=& -p_j c(a_k, a_j, x)+ p_k c(a_j, a_k, x) - \sum_{l: l \in \mathcal{A} \setminus\{a_j, a_k\} } p_l \lbc c(a_k, a_l, x) - c(a_j, a_l, x) \rbc.
\end{eqnarray*}
 \end{proof}

\begin{proof}[Proof of Theorem~\ref{thm: Ycorrection}]
The causal effect is defined as the difference between the potential outcomes and the estimates of causal effects are based on the observed outcomes.  To correct the potential bias, we adjust the observed outcome $Y$ of an individual who received treatment $a_j$ as follows
\begin{eqnarray*}
\YCF &=& Y - \lbc \E \lsq  Y(a_j) \cond a_j, x\rsq -  \E \lsq  Y(a_j) \cond x\rsq \rbc. 
\end{eqnarray*}
Let $p_l \equiv \mathbb{P}(A=a_l \cond X=x)$ and $\E [\cdot \cond a_l,x] \equiv \E [\cdot \cond A=a_l,X=x], \forall a_l \in \mathcal{A}$.  By applying the law of total expectation to $\E \lsq  Y(a_j) \cond x\rsq$, we can show that 
\begin{eqnarray*}
 &&\E \lsq  Y(a_j) \cond a_j, x\rsq -  \E \lsq  Y(a_j) \cond x\rsq \\
 &=& \E \lsq  Y(a_j) \cond a_j, x\rsq  -  \sum_{l=1}^J p_l\E \lsq  Y(a_j) \cond a_l, x\rsq \\
 &=&(1-p_j) \E \lsq  Y(a_j) \cond a_j, x\rsq - \sum_{l\neq j}^J p_l \E \lsq  Y(a_j) \cond a_l, x\rsq\\
 &=& \sum_{l \neq j}^J p_l \lbc \E \lsq  Y(a_j) \cond a_j, x\rsq -\E \lsq  Y(a_j) \cond a_l, x\rsq  \rbc\\
 &=&  \sum_{l \neq j}^J p_l c(a_j, a_l,x).
\end{eqnarray*}
Now we prove that replacing $Y$ with $\YCF$ removes the bias in equation~\eqref{eq:biasform} of Theorem~\ref{thm:biasform}. 
Consider the causal effect between any pair of treatments $a_j$ and $a_k$. Using the adjusted outcomes $\YCF$, the estimate of the causal effect is 
\begin{eqnarray*}
&&\E \lsq  \YCF \cond a_j, x  \rsq -\E \lsq  \YCF \cond a_j, x  \rsq \\
&=& \E \lsq \lp Y- \sum_{l\neq j}^J p_l c(a_j, a_l, x) \rp \cond a_j, x \rsq -  \E \lsq \lp Y- \sum_{l\neq k}^J p_l c(a_k, a_l, x) \rp \cond a_k, x\rsq\\
&=& \E (Y \cond a_j, x) - \E (Y \cond a_k, x) \underbrace{+p_j c(a_k, a_j, x) - p_k c(a_j, a_k, x) + \sum\limits_{l: l \in \mathcal{A}\setminus\{a_j, a_k\}} p_{l} \lbc c(a_k, a_l, x) -c(a_j,a_l,x) \rbc}_{\text{--Bias in equation~\eqref{eq:biasform}}}.
\end{eqnarray*} 
\end{proof}
\end{appendix}

\section*{Acknowledgements}

The first author was supported in part by NIH Grants  R21 CA245855-01and and P30CA196521-01, and by award ME\_2017C3\_9041 from the Patient-Centered Outcomes Research Institute.



\bibliographystyle{imsart-nameyear} 
\bibliography{references}       

\begin{thebibliography}{42}

\bibitem[\protect\citeauthoryear{Bill{\'e}
  et~al.}{2021}]{bille2021preoperative}
\begin{barticle}[author]
\bauthor{\bsnm{Bill{\'e}},~\bfnm{Andrea}\binits{A.}},
  \bauthor{\bsnm{Buxton},~\bfnm{James}\binits{J.}},
  \bauthor{\bsnm{Viviano},~\bfnm{Alessandro}\binits{A.}},
  \bauthor{\bsnm{Gammon},~\bfnm{David}\binits{D.}},
  \bauthor{\bsnm{Veres},~\bfnm{Lukacs}\binits{L.}},
  \bauthor{\bsnm{Routledge},~\bfnm{Tom}\binits{T.}},
  \bauthor{\bsnm{Harrison-Phipps},~\bfnm{Karen}\binits{K.}},
  \bauthor{\bsnm{Dixon},~\bfnm{Allison}\binits{A.}} \AND
  \bauthor{\bsnm{Minetto},~\bfnm{Marco~A}\binits{M.~A.}}
(\byear{2021}).
\btitle{Preoperative Physical Activity Predicts Surgical Outcomes Following
  Lung Cancer Resection}.
\bjournal{Integrative Cancer Therapies}
\bvolume{20}
\bpages{1--8}.
\end{barticle}
\endbibitem

\bibitem[\protect\citeauthoryear{Brumback
  et~al.}{2004}]{brumback2004sensitivity}
\begin{barticle}[author]
\bauthor{\bsnm{Brumback},~\bfnm{Babette~A}\binits{B.~A.}},
  \bauthor{\bsnm{Hern{\'a}n},~\bfnm{Miguel~A}\binits{M.~A.}},
  \bauthor{\bsnm{Haneuse},~\bfnm{Sebastien~JPA}\binits{S.~J.}} \AND
  \bauthor{\bsnm{Robins},~\bfnm{James~M}\binits{J.~M.}}
(\byear{2004}).
\btitle{Sensitivity analyses for unmeasured confounding assuming a marginal
  structural model for repeated measures}.
\bjournal{Statistics in Medicine}
\bvolume{23}
\bpages{749--767}.
\end{barticle}
\endbibitem

\bibitem[\protect\citeauthoryear{Ceppa et~al.}{2012}]{ceppa2012thoracoscopic}
\begin{barticle}[author]
\bauthor{\bsnm{Ceppa},~\bfnm{DuyKhanh~P}\binits{D.~P.}},
  \bauthor{\bsnm{Kosinski},~\bfnm{Andrzej~S}\binits{A.~S.}},
  \bauthor{\bsnm{Berry},~\bfnm{Mark~F}\binits{M.~F.}},
  \bauthor{\bsnm{Tong},~\bfnm{Betty~C}\binits{B.~C.}},
  \bauthor{\bsnm{Harpole},~\bfnm{David~H}\binits{D.~H.}},
  \bauthor{\bsnm{Mitchell},~\bfnm{John~D}\binits{J.~D.}},
  \bauthor{\bsnm{D'Amico},~\bfnm{Thomas~A}\binits{T.~A.}} \AND
  \bauthor{\bsnm{Onaitis},~\bfnm{Mark~W}\binits{M.~W.}}
(\byear{2012}).
\btitle{{Thoracoscopic lobectomy has increasing benefit in patients with poor
  pulmonary function: a Society of Thoracic Surgeons Database analysis}}.
\bjournal{Annals of Surgery}
\bvolume{256}
\bpages{487}.
\end{barticle}
\endbibitem

\bibitem[\protect\citeauthoryear{Chipman et~al.}{2010}]{chipman2010bart}
\begin{barticle}[author]
\bauthor{\bsnm{Chipman},~\bfnm{Hugh~A}\binits{H.~A.}},
  \bauthor{\bsnm{George},~\bfnm{Edward~I}\binits{E.~I.}},
  \bauthor{\bsnm{McCulloch},~\bfnm{Robert~E}\binits{R.~E.}} \betal{et~al.}
(\byear{2010}).
\btitle{{BART}: Bayesian additive regression trees}.
\bjournal{The Annals of Applied Statistics}
\bvolume{4}
\bpages{266--298}.
\end{barticle}
\endbibitem

\bibitem[\protect\citeauthoryear{Daniels and Hogan}{2008}]{daniels2008missing}
\begin{bbook}[author]
\bauthor{\bsnm{Daniels},~\bfnm{Michael~J}\binits{M.~J.}} \AND
  \bauthor{\bsnm{Hogan},~\bfnm{Joseph~W}\binits{J.~W.}}
(\byear{2008}).
\btitle{Missing data in longitudinal studies: Strategies for Bayesian modeling
  and sensitivity analysis}.
\bpublisher{Boca Raton, FL: CRC Press}.
\end{bbook}
\endbibitem

\bibitem[\protect\citeauthoryear{Ding and
  VanderWeele}{2016}]{ding2016sensitivity}
\begin{barticle}[author]
\bauthor{\bsnm{Ding},~\bfnm{Peng}\binits{P.}} \AND
  \bauthor{\bsnm{VanderWeele},~\bfnm{Tyler~J}\binits{T.~J.}}
(\byear{2016}).
\btitle{Sensitivity analysis without assumptions}.
\bjournal{Epidemiology}
\bvolume{27}
\bpages{368}.
\end{barticle}
\endbibitem

\bibitem[\protect\citeauthoryear{Dorie et~al.}{2016}]{dorie2016flexible}
\begin{barticle}[author]
\bauthor{\bsnm{Dorie},~\bfnm{Vincent}\binits{V.}},
  \bauthor{\bsnm{Harada},~\bfnm{Masataka}\binits{M.}},
  \bauthor{\bsnm{Carnegie},~\bfnm{Nicole~Bohme}\binits{N.~B.}} \AND
  \bauthor{\bsnm{Hill},~\bfnm{Jennifer}\binits{J.}}
(\byear{2016}).
\btitle{A flexible, interpretable framework for assessing sensitivity to
  unmeasured confounding}.
\bjournal{Statistics in Medicine}
\bvolume{35}
\bpages{3453--3470}.
\end{barticle}
\endbibitem

\bibitem[\protect\citeauthoryear{Greenland}{2005}]{greenland2005multiple}
\begin{barticle}[author]
\bauthor{\bsnm{Greenland},~\bfnm{Sander}\binits{S.}}
(\byear{2005}).
\btitle{Multiple-bias modelling for analysis of observational data}.
\bjournal{Journal of the Royal Statistical Society: Series A (Statistics in
  Society)}
\bvolume{168}
\bpages{267--306}.
\end{barticle}
\endbibitem

\bibitem[\protect\citeauthoryear{Gu and Gutman}{2019}]{gu2019development}
\begin{barticle}[author]
\bauthor{\bsnm{Gu},~\bfnm{Chenyang}\binits{C.}} \AND
  \bauthor{\bsnm{Gutman},~\bfnm{Roee}\binits{R.}}
(\byear{2019}).
\btitle{Development of a common patient assessment scale across the continuum
  of care: a nested multiple imputation approach}.
\bjournal{The Annals of Applied Statistics}
\bvolume{13}
\bpages{466--491}.
\end{barticle}
\endbibitem

\bibitem[\protect\citeauthoryear{Gustafson and
  McCandless}{2018}]{gustafson2018}
\begin{barticle}[author]
\bauthor{\bsnm{Gustafson},~\bfnm{Paul}\binits{P.}} \AND
  \bauthor{\bsnm{McCandless},~\bfnm{Lawrence~C}\binits{L.~C.}}
(\byear{2018}).
\btitle{When Is a Sensitivity Parameter Exactly That?}
\bjournal{Statistical Science}
\bvolume{33}
\bpages{86--95}.
\end{barticle}
\endbibitem

\bibitem[\protect\citeauthoryear{Hahn, Murray and
  Carvalho}{2020}]{hahn2020bayesian}
\begin{barticle}[author]
\bauthor{\bsnm{Hahn},~\bfnm{P~Richard}\binits{P.~R.}},
  \bauthor{\bsnm{Murray},~\bfnm{Jared~S}\binits{J.~S.}} \AND
  \bauthor{\bsnm{Carvalho},~\bfnm{Carlos~M}\binits{C.~M.}}
(\byear{2020}).
\btitle{Bayesian regression tree models for causal inference: regularization,
  confounding, and heterogeneous effects (with Discussion)}.
\bjournal{Bayesian Analysis}
\bvolume{15}
\bpages{965-1056}.
\bdoi{10.1214/19-BA1195.}
\end{barticle}
\endbibitem

\bibitem[\protect\citeauthoryear{Hill}{2011}]{hill2011bayesian}
\begin{barticle}[author]
\bauthor{\bsnm{Hill},~\bfnm{Jennifer~L}\binits{J.~L.}}
(\byear{2011}).
\btitle{Bayesian nonparametric modeling for causal inference}.
\bjournal{Journal of Computational and Graphical Statistics}
\bvolume{20}
\bpages{217--240}.
\end{barticle}
\endbibitem

\bibitem[\protect\citeauthoryear{Hogan, Daniels and
  Hu}{2014}]{hogan2014bayesian}
\begin{bincollection}[author]
\bauthor{\bsnm{Hogan},~\bfnm{Joseph~W.}\binits{J.~W.}},
  \bauthor{\bsnm{Daniels},~\bfnm{Michael~J.}\binits{M.~J.}} \AND
  \bauthor{\bsnm{Hu},~\bfnm{Liangyuan}\binits{L.}}
(\byear{2014}).
\btitle{A Bayesian perspective on assessing sensitivity to assumptions about
  unobserved data}.
In \bbooktitle{Handbook of Missing Data Methodology}
(\beditor{\bfnm{Geert}\binits{G.}~\bsnm{Molenberghs}},
  \beditor{\bfnm{Garrett}\binits{G.}~\bsnm{Fitzmaurice}},
  \beditor{\bfnm{Michael~G.}\binits{M.~G.}~\bsnm{Kenward}},
  \beditor{\bfnm{Anastasios}\binits{A.}~\bsnm{Tsiatis}} \AND
  \beditor{\bfnm{Geert}\binits{G.}~\bsnm{Verbeke}}, eds.)
\bchapter{18},
\bpages{405--434}.
\bpublisher{Boca Raton, FL: CRC Press}.
\end{bincollection}
\endbibitem

\bibitem[\protect\citeauthoryear{Howington
  et~al.}{2013}]{howington2013treatment}
\begin{barticle}[author]
\bauthor{\bsnm{Howington},~\bfnm{John~A}\binits{J.~A.}},
  \bauthor{\bsnm{Blum},~\bfnm{Matthew~G}\binits{M.~G.}},
  \bauthor{\bsnm{Chang},~\bfnm{Andrew~C}\binits{A.~C.}},
  \bauthor{\bsnm{Balekian},~\bfnm{Alex~A}\binits{A.~A.}} \AND
  \bauthor{\bsnm{Murthy},~\bfnm{Sudish~C}\binits{S.~C.}}
(\byear{2013}).
\btitle{{Treatment of stage I and II non-small cell lung cancer: diagnosis and
  management of lung cancer: American College of Chest Physicians
  evidence-based clinical practice guidelines}}.
\bjournal{Chest}
\bvolume{143}
\bpages{e278S--e313S}.
\end{barticle}
\endbibitem

\bibitem[\protect\citeauthoryear{Hu}{2020}]{hucomment2020}
\begin{barticle}[author]
\bauthor{\bsnm{Hu},~\bfnm{Liangyuan}\binits{L.}}
(\byear{2020}).
\btitle{{Discussion on “Bayesian Regression Tree Models for Causal Inference:
  Regularization, Confounding, and Heterogeneous Effects” by Hahn, Murray and
  Carvalho}}.
\bjournal{Bayesian Analysis}
\bvolume{15}
\bpages{1020--1023}.
\end{barticle}
\endbibitem

\bibitem[\protect\citeauthoryear{Hu and Gu}{2020}]{hu2020rare}
\begin{barticle}[author]
\bauthor{\bsnm{Hu},~\bfnm{Liangyuan}\binits{L.}} \AND
  \bauthor{\bsnm{Gu},~\bfnm{Chenyang}\binits{C.}}
(\byear{2020}).
\btitle{Estimation of causal effects of multiple treatments in healthcare
  database studies with rare outcomes}.
\bjournal{Health Services and Outcomes Research Methodology}
\bvolume{21}
\bpages{287--308}.
\end{barticle}
\endbibitem

\bibitem[\protect\citeauthoryear{Hu and Hogan}{2019}]{hu2019causal}
\begin{barticle}[author]
\bauthor{\bsnm{Hu},~\bfnm{Liangyuan}\binits{L.}} \AND
  \bauthor{\bsnm{Hogan},~\bfnm{Joseph~W}\binits{J.~W.}}
(\byear{2019}).
\btitle{Causal comparative effectiveness analysis of dynamic continuous-time
  treatment initiation rules with sparsely measured outcomes and death}.
\bjournal{Biometrics}
\bvolume{75}
\bpages{695--707}.
\end{barticle}
\endbibitem

\bibitem[\protect\citeauthoryear{Hu, Ji and Li}{2021}]{hu2021estimatinghet}
\begin{barticle}[author]
\bauthor{\bsnm{Hu},~\bfnm{Liangyuan}\binits{L.}},
  \bauthor{\bsnm{Ji},~\bfnm{Jiayi}\binits{J.}} \AND
  \bauthor{\bsnm{Li},~\bfnm{Fan}\binits{F.}}
(\byear{2021}).
\btitle{Estimating heterogeneous survival treatment effect in observational
  data using machine learning}.
\bjournal{Statistics in Medicine}
\bpages{In press}.
\end{barticle}
\endbibitem

\bibitem[\protect\citeauthoryear{Hu, Lin and Ji}{2021}]{hu2021variable}
\begin{barticle}[author]
\bauthor{\bsnm{Hu},~\bfnm{Liangyuan}\binits{L.}},
  \bauthor{\bsnm{Lin},~\bfnm{Jung-Yi~Joyce}\binits{J.-Y.~J.}} \AND
  \bauthor{\bsnm{Ji},~\bfnm{Jiayi}\binits{J.}}
(\byear{2021}).
\btitle{Variable selection with missing data in both covariates and outcomes:
  Imputation and machine learning}.
\bjournal{arXiv preprint arXiv:2104.02769}.
\end{barticle}
\endbibitem

\bibitem[\protect\citeauthoryear{Hu, Liu and Li}{2020}]{hu2020ranking}
\begin{barticle}[author]
\bauthor{\bsnm{Hu},~\bfnm{Liangyuan}\binits{L.}},
  \bauthor{\bsnm{Liu},~\bfnm{Bian}\binits{B.}} \AND
  \bauthor{\bsnm{Li},~\bfnm{Yan}\binits{Y.}}
(\byear{2020}).
\btitle{Ranking sociodemographic, health behavior, prevention, and
  environmental factors in predicting neighborhood cardiovascular health: {A
  Bayesian} machine learning approach}.
\bjournal{Preventive Medicine}
\bvolume{141}
\bpages{106240}.
\end{barticle}
\endbibitem

\bibitem[\protect\citeauthoryear{Hu et~al.}{2018}]{hu2017modeling}
\begin{barticle}[author]
\bauthor{\bsnm{Hu},~\bfnm{Liangyuan}\binits{L.}},
  \bauthor{\bsnm{Hogan},~\bfnm{Joseph~W}\binits{J.~W.}},
  \bauthor{\bsnm{Mwangi},~\bfnm{Ann~W}\binits{A.~W.}} \AND
  \bauthor{\bsnm{Siika},~\bfnm{Abraham}\binits{A.}}
(\byear{2018}).
\btitle{Modeling the causal effect of treatment initiation time on survival:
  Application to {HIV/TB} co-infection}.
\bjournal{Biometrics}
\bvolume{74}
\bpages{703--713}.
\end{barticle}
\endbibitem

\bibitem[\protect\citeauthoryear{Hu et~al.}{2020a}]{hu2020estimation}
\begin{barticle}[author]
\bauthor{\bsnm{Hu},~\bfnm{Liangyuan}\binits{L.}},
  \bauthor{\bsnm{Gu},~\bfnm{Chenyang}\binits{C.}},
  \bauthor{\bsnm{Lopez},~\bfnm{Michael}\binits{M.}},
  \bauthor{\bsnm{Ji},~\bfnm{Jiayi}\binits{J.}} \AND
  \bauthor{\bsnm{Wisnivesky},~\bfnm{Juan}\binits{J.}}
(\byear{2020}a).
\btitle{Estimation of causal effects of multiple treatments in observational
  studies with a binary outcome}.
\bjournal{Statistical Methods in Medical Research}
\bvolume{29}
\bpages{3218--3234}.
\end{barticle}
\endbibitem

\bibitem[\protect\citeauthoryear{Hu et~al.}{2020b}]{hu2020tree}
\begin{barticle}[author]
\bauthor{\bsnm{Hu},~\bfnm{Liangyuan}\binits{L.}},
  \bauthor{\bsnm{Liu},~\bfnm{Bian}\binits{B.}},
  \bauthor{\bsnm{Ji},~\bfnm{Jiayi}\binits{J.}} \AND
  \bauthor{\bsnm{Li},~\bfnm{Yan}\binits{Y.}}
(\byear{2020}b).
\btitle{{Tree-Based Machine Learning to Identify and Understand Major
  Determinants for Stroke at the Neighborhood Level}}.
\bjournal{Journal of the American Heart Association}
\bvolume{9}
\bpages{e016745}.
\end{barticle}
\endbibitem

\bibitem[\protect\citeauthoryear{Hu et~al.}{2021}]{hu2021estimatingaoe}
\begin{barticle}[author]
\bauthor{\bsnm{Hu},~\bfnm{Liangyuan}\binits{L.}},
  \bauthor{\bsnm{Lin},~\bfnm{Jung-Yi~Joyce}\binits{J.-Y.~J.}},
  \bauthor{\bsnm{Sigel},~\bfnm{Keith}\binits{K.}} \AND
  \bauthor{\bsnm{Kale},~\bfnm{Minal}\binits{M.}}
(\byear{2021}).
\btitle{Estimating heterogeneous survival treatment effects of lung cancer
  screening approaches: A causal machine learning analysis}.
\bjournal{Annals of Epidemiology}
\bvolume{62}
\bpages{36--42}.
\end{barticle}
\endbibitem

\bibitem[\protect\citeauthoryear{Imbens}{2003}]{imbens2003sensitivity}
\begin{barticle}[author]
\bauthor{\bsnm{Imbens},~\bfnm{Guido~W}\binits{G.~W.}}
(\byear{2003}).
\btitle{Sensitivity to exogeneity assumptions in program evaluation}.
\bjournal{American Economic Review}
\bvolume{93}
\bpages{126--132}.
\end{barticle}
\endbibitem

\bibitem[\protect\citeauthoryear{Kasza, Wolfe and
  Schuster}{2017}]{kasza2017assessing}
\begin{barticle}[author]
\bauthor{\bsnm{Kasza},~\bfnm{Jessica}\binits{J.}},
  \bauthor{\bsnm{Wolfe},~\bfnm{Rory}\binits{R.}} \AND
  \bauthor{\bsnm{Schuster},~\bfnm{Tibor}\binits{T.}}
(\byear{2017}).
\btitle{Assessing the impact of unmeasured confounding for binary outcomes
  using confounding functions}.
\bjournal{International Journal of Epidemiology}
\bvolume{46}
\bpages{1303--1311}.
\end{barticle}
\endbibitem

\bibitem[\protect\citeauthoryear{Lakens}{2013}]{lakens2013calculating}
\begin{barticle}[author]
\bauthor{\bsnm{Lakens},~\bfnm{Dani{\"e}l}\binits{D.}}
(\byear{2013}).
\btitle{Calculating and reporting effect sizes to facilitate cumulative
  science: a practical primer for t-tests and ANOVAs}.
\bjournal{Frontiers in Psychology}
\bvolume{4}
\bpages{863}.
\end{barticle}
\endbibitem

\bibitem[\protect\citeauthoryear{Lash, Fox and Fink}{2011}]{lash2011applying}
\begin{bbook}[author]
\bauthor{\bsnm{Lash},~\bfnm{Timothy~L}\binits{T.~L.}},
  \bauthor{\bsnm{Fox},~\bfnm{Matthew~P}\binits{M.~P.}} \AND
  \bauthor{\bsnm{Fink},~\bfnm{Aliza~K}\binits{A.~K.}}
(\byear{2011}).
\btitle{Applying Quantitative Bias Analysis to Epidemiologic Data}.
\bpublisher{New York: Springer Science \& Business Media}.
\end{bbook}
\endbibitem

\bibitem[\protect\citeauthoryear{Li et~al.}{2011}]{li2011propensity}
\begin{barticle}[author]
\bauthor{\bsnm{Li},~\bfnm{Lingling}\binits{L.}},
  \bauthor{\bsnm{Shen},~\bfnm{Changyu}\binits{C.}},
  \bauthor{\bsnm{Wu},~\bfnm{Ann~C}\binits{A.~C.}} \AND
  \bauthor{\bsnm{Li},~\bfnm{Xiaochun}\binits{X.}}
(\byear{2011}).
\btitle{Propensity score-based sensitivity analysis method for uncontrolled
  confounding}.
\bjournal{American Journal of Epidemiology}
\bvolume{174}
\bpages{345--353}.
\end{barticle}
\endbibitem

\bibitem[\protect\citeauthoryear{Lin, Psaty and
  Kronmal}{1998}]{lin1998assessing}
\begin{barticle}[author]
\bauthor{\bsnm{Lin},~\bfnm{Danyu~Y}\binits{D.~Y.}},
  \bauthor{\bsnm{Psaty},~\bfnm{Bruce~M}\binits{B.~M.}} \AND
  \bauthor{\bsnm{Kronmal},~\bfnm{Richard~A}\binits{R.~A.}}
(\byear{1998}).
\btitle{Assessing the sensitivity of regression results to unmeasured
  confounders in observational studies}.
\bjournal{Biometrics}
\bvolume{54}
\bpages{948--963}.
\end{barticle}
\endbibitem

\bibitem[\protect\citeauthoryear{McCandless and
  Gustafson}{2017}]{mccandless2017comparison}
\begin{barticle}[author]
\bauthor{\bsnm{McCandless},~\bfnm{Lawrence~C}\binits{L.~C.}} \AND
  \bauthor{\bsnm{Gustafson},~\bfnm{Paul}\binits{P.}}
(\byear{2017}).
\btitle{A comparison of Bayesian and Monte Carlo sensitivity analysis for
  unmeasured confounding}.
\bjournal{Statistics in Medicine}
\bvolume{36}
\bpages{2887--2901}.
\end{barticle}
\endbibitem

\bibitem[\protect\citeauthoryear{Robins}{1999}]{robins1999association}
\begin{barticle}[author]
\bauthor{\bsnm{Robins},~\bfnm{James~M}\binits{J.~M.}}
(\byear{1999}).
\btitle{Association, causation, and marginal structural models}.
\bjournal{Synthese}
\bvolume{121}
\bpages{151--179}.
\end{barticle}
\endbibitem

\bibitem[\protect\citeauthoryear{Rosenbaum and
  Rubin}{1983}]{rosenbaum1983assessing}
\begin{barticle}[author]
\bauthor{\bsnm{Rosenbaum},~\bfnm{Paul~R}\binits{P.~R.}} \AND
  \bauthor{\bsnm{Rubin},~\bfnm{Donald~B}\binits{D.~B.}}
(\byear{1983}).
\btitle{Assessing sensitivity to an unobserved binary covariate in an
  observational study with binary outcome}.
\bjournal{Journal of the Royal Statistical Society: Series B (Methodological)}
\bvolume{45}
\bpages{212--218}.
\end{barticle}
\endbibitem

\bibitem[\protect\citeauthoryear{Ruan and Kulkarni}{2020}]{VATS5336}
\begin{barticle}[author]
\bauthor{\bsnm{Ruan},~\bfnm{Alexandra}\binits{A.}} \AND
  \bauthor{\bsnm{Kulkarni},~\bfnm{Vivek}\binits{V.}}
(\byear{2020}).
\btitle{Anesthesia considerations for robotic thoracic surgery}.
\bjournal{{Video-Assisted Thoracic Surgery}}
\bvolume{5}
\bpages{1--8}.
\end{barticle}
\endbibitem

\bibitem[\protect\citeauthoryear{Rubin}{1974}]{rubin1974estimating}
\begin{barticle}[author]
\bauthor{\bsnm{Rubin},~\bfnm{Donald~B}\binits{D.~B.}}
(\byear{1974}).
\btitle{Estimating causal effects of treatments in randomized and nonrandomized
  studies.}
\bjournal{Journal of Educational Psychology}
\bvolume{66}
\bpages{688}.
\end{barticle}
\endbibitem

\bibitem[\protect\citeauthoryear{Rubin}{2003}]{rubin2003nested}
\begin{barticle}[author]
\bauthor{\bsnm{Rubin},~\bfnm{Donald~B}\binits{D.~B.}}
(\byear{2003}).
\btitle{Nested multiple imputation of NMES via partially incompatible MCMC}.
\bjournal{Statistica Neerlandica}
\bvolume{57}
\bpages{3--18}.
\end{barticle}
\endbibitem

\bibitem[\protect\citeauthoryear{Saito et~al.}{2017}]{saito2017impact}
\begin{barticle}[author]
\bauthor{\bsnm{Saito},~\bfnm{Hajime}\binits{H.}},
  \bauthor{\bsnm{Hatakeyama},~\bfnm{Kazutoshi}\binits{K.}},
  \bauthor{\bsnm{Konno},~\bfnm{Hayato}\binits{H.}},
  \bauthor{\bsnm{Matsunaga},~\bfnm{Toshiki}\binits{T.}},
  \bauthor{\bsnm{Shimada},~\bfnm{Yoichi}\binits{Y.}} \AND
  \bauthor{\bsnm{Minamiya},~\bfnm{Yoshihiro}\binits{Y.}}
(\byear{2017}).
\btitle{Impact of pulmonary rehabilitation on postoperative complications in
  patients with lung cancer and chronic obstructive pulmonary disease}.
\bjournal{Thoracic Cancer}
\bvolume{8}
\bpages{451--460}.
\end{barticle}
\endbibitem

\bibitem[\protect\citeauthoryear{Sihoe}{2020}]{sihoe2020video}
\begin{barticle}[author]
\bauthor{\bsnm{Sihoe},~\bfnm{Alan~DL}\binits{A.~D.}}
(\byear{2020}).
\btitle{{Video-assisted thoracoscopic surgery as the gold standard for lung
  cancer surgery}}.
\bjournal{Respirology}
\bvolume{25}
\bpages{49--60}.
\end{barticle}
\endbibitem

\bibitem[\protect\citeauthoryear{VanderWeele and
  Arah}{2011}]{vanderweele2011unmeasured}
\begin{barticle}[author]
\bauthor{\bsnm{VanderWeele},~\bfnm{Tyler~J}\binits{T.~J.}} \AND
  \bauthor{\bsnm{Arah},~\bfnm{Onyebuchi~A}\binits{O.~A.}}
(\byear{2011}).
\btitle{Unmeasured confounding for general outcomes, treatments, and
  confounders: bias formulas for sensitivity analysis}.
\bjournal{Epidemiology}
\bvolume{22}
\bpages{42}.
\end{barticle}
\endbibitem

\bibitem[\protect\citeauthoryear{Von~Elm et~al.}{2007}]{von2007strengthening}
\begin{barticle}[author]
\bauthor{\bsnm{Von~Elm},~\bfnm{Erik}\binits{E.}},
  \bauthor{\bsnm{Altman},~\bfnm{Douglas~G}\binits{D.~G.}},
  \bauthor{\bsnm{Egger},~\bfnm{Matthias}\binits{M.}},
  \bauthor{\bsnm{Pocock},~\bfnm{Stuart~J}\binits{S.~J.}},
  \bauthor{\bsnm{G{\o}tzsche},~\bfnm{Peter~C}\binits{P.~C.}} \AND
  \bauthor{\bsnm{Vandenbroucke},~\bfnm{Jan~P}\binits{J.~P.}}
(\byear{2007}).
\btitle{The Strengthening the Reporting of Observational Studies in
  Epidemiology (STROBE) statement: guidelines for reporting observational
  studies}.
\bjournal{Annals of Internal Medicine}
\bvolume{147}
\bpages{573--577}.
\end{barticle}
\endbibitem

\bibitem[\protect\citeauthoryear{Zhou and Reiter}{2010}]{zhou2010note}
\begin{barticle}[author]
\bauthor{\bsnm{Zhou},~\bfnm{Xiang}\binits{X.}} \AND
  \bauthor{\bsnm{Reiter},~\bfnm{Jerome~P}\binits{J.~P.}}
(\byear{2010}).
\btitle{{A note on Bayesian inference after multiple imputation}}.
\bjournal{The American Statistician}
\bvolume{64}
\bpages{159--163}.
\end{barticle}
\endbibitem

\bibitem[\protect\citeauthoryear{Zigler et~al.}{2013}]{zigler2013model}
\begin{barticle}[author]
\bauthor{\bsnm{Zigler},~\bfnm{Corwin~M}\binits{C.~M.}},
  \bauthor{\bsnm{Watts},~\bfnm{Krista}\binits{K.}},
  \bauthor{\bsnm{Yeh},~\bfnm{Robert~W}\binits{R.~W.}},
  \bauthor{\bsnm{Wang},~\bfnm{Yun}\binits{Y.}},
  \bauthor{\bsnm{Coull},~\bfnm{Brent~A}\binits{B.~A.}} \AND
  \bauthor{\bsnm{Dominici},~\bfnm{Francesca}\binits{F.}}
(\byear{2013}).
\btitle{Model feedback in Bayesian propensity score estimation}.
\bjournal{Biometrics}
\bvolume{69}
\bpages{263--273}.
\end{barticle}
\endbibitem

\end{thebibliography}

\end{document}


\begin{frontmatter}
\title{Supplementary Materials for ``A flexible sensitivity analysis approach for unmeasured confounding with multiple treatments and a binary outcome with application to SEER-Medicare lung cancer data''\thanksref{t1}}
\runtitle{Sensitivity analysis with multiple treatments}
\thankstext{T1}{Corresponding Author: Liangyuan Hu,  Department of Biostatistics and Epidemiology, Rutgers School of Public Health, Piscataway, NJ 08854, USA. Email:liangyuan.hu@rutgers.edu}

\begin{aug}
\author[A]{\fnms{Liangyuan} \snm{Hu}\ead[label=e1]{liangyuan.hu@rutgers.edu}},
\author[B]{\fnms{Jungang} \snm{Zou}\ead[label=e2]{jz3183@cumc.columbia.edu}},
\author[C]{\fnms{Chenyang} \snm{Gu}\ead[label=e3]{chenyang.gu@analysisgroup.com}},
\author[D]{\fnms{Jiayi} \snm{Ji}\ead[label=e4]{jiayi.ji@mountsinai.org}},
\author[E]{\fnms{Michael} \snm{Lopez}\ead[label=e5]{mlopez1@skidmore.edu}}
\and
\author[F]{\fnms{Minal} \snm{Kale}\ead[label=e6]{minal.kale@mountsinai.org}}
\address[A]{Department of Biostatistics and Epidemiology,
Rutgers School of Public Health,
\printead{e1}}

\address[B]{Department of Biostatistics,
Columbia University,
\printead{e2}}

\address[C]{Analysis Group, Inc.,
\printead{e3}}

\address[D]{Department of Population Health Science and Policy,
Icahn School of Medicine at Mount Sinai,
\printead{e4}}

\address[E]{Department of Mathematics, Skidmore College,
\printead{e5}}

\address[F]{Department of Medicine,
Icahn School of Medicine at Mount Sinai,
\printead{e6}}

\end{aug}

\end{frontmatter}


\section{Design of supplementary simulation in Section 3.1}
In a simplified scenario of the illustrative simulation in Section 3.1, we assume the independence between $X_1$ and $U$. We used the same treatment assignment model and modified the outcome generating models to maintain the same observed outcome event probabilites in each of three treatment groups.  Three sets of nonparallel response surfaces were generated,
\begin{align*}
\mathbb{P} \lp Y (1) =1 \mid X_1, U \rp& =\text{logit}^{-1}(0.1X_1-1.8U)\\
\mathbb{P} \lp Y (2) =1 \mid X_1, U \rp & =\text{logit}^{-1}(-0.7X_1+1.6U)\\
\mathbb{P} \lp Y (3) =1 \mid X_1, U \rp & =\text{logit}^{-1}(-0.5X_1+2.1U). 
\end{align*}

Under this simulation configuration, the observed outcome event probability was 0.40 in $A=1$, 0.51 in $A=2$ and 0.64 in $A=3$ and the true $\text{CATE}_{1,2}=-0.16$, $\text{CATE}_{1,3}=-0.29$ and $\text{CATE}_{2,3}=-0.13$. Web Figure 1 shows the estimates of three pairwise causal effects among 1000 replications corresponding to each of four strategies (I)-(IV) described in Section 3.1. \\

\section{Supplementary tables and figures} 
Web tables and figures referenced in the paper are provided below. 

\begin{table}[H]
\centering
\caption{The data generating process of the covariates.}
\begin{tabular}{cccccc}\hline
Variables &  Distribution &Variables &  Distribution &Variables &  Distribution \\\hline
$X_1$ & $N(0,1)$ & $X_6$ &$\text{Bern}(0.6)$ & $X_{11}$ & Student's $t_{10}$ \\
$X_2$ & $\mathcal{U}(-1,1)$ &$X_7$ & $\text{Bern}(0.3)$ &$X_{12}$ & Gamma(2,2)\\
$X_3$ &Beta(3,3)&$X_8$ & $\text{Bern}(0.5)$ &$X_{13}$ & InverseGamma(20,20)\\
$X_4$ &$N(-1,1)$ & $X_9$ & $\text{Multinom}(1, 0.3, 0.2,0.5)$ & $X_{14}$& $N(-1,2)$\\
$X_5$ &$N(1,1)$ &$X_{10}$&$\text{Multinom}(1, 0.1, 0.8,0.1)$ & $X_{15}$ &$N(1,2)$\\\hline
\end{tabular}
\end{table}

\setlength{\tabcolsep}{3pt} 
\begin{table}[H]
\centering
\caption{Average absolute bias (AAB) and root-mean-squared error (RMSE) in the estimated conditional average causal effects for illustrative simulation in Section 3.1. Sensitivity analysis strategies (I)-(IV) were used. I: True $c^0$. I: 3rd $A$ ignored: $c(\cdot)$ functions involving the third treatment were set to zero. II: $\mathcal{U} \lp \max(-1, c^0 - h\hat{\sigma}), \min(1, c^0+h\hat{\sigma}) \rp$. III: $\mathcal{U} \lp \max(-1, c^0 - 2h \hat{\sigma}),  c^0 \rp$  or $\mathcal{U} \lp  c^0 , \min(1, c^0+2h \hat{\sigma}) \rp$. IV: $\mathcal{U}(-1,1)$. The CATE results that could be achieved if $U$ were actually observed and the naive CATE estimators ignoring $U$  are also presented. The true $CATE_{1,2}=-0.16$, $CATE_{1,3}=-0.29$ and $CATE_{2,3}=-0.13$. }
\small
\begin{tabular}{cccccccccccccc}\hline
&\multicolumn{6}{c}{$c(a_1,a_2,x_1)$}&&\multicolumn{6}{c}{$c(a_1,a_2)$}\\\cline{2-14}
&\multicolumn{2}{c}{$CATE_{1,2}$}&\multicolumn{2}{c}{$CATE_{1,3}$}&\multicolumn{2}{c}{$CATE_{2,3}$}&&\multicolumn{2}{c}{$CATE_{1,2}$}&\multicolumn{2}{c}{$CATE_{1,3}$}&\multicolumn{2}{c}{$CATE_{2,3}$}\\
& AAB&RMSE&AAB&RMSE&AAB&RMSE&&AAB&RMSE&AAB&RMSE&AAB&RMSE\\\hline
$U$ included&.01&.01&.01&.01&.01&.01&&.01&.01&.01&.01&.01&.01\\
I &.01&.01&.01&.01&.01&.01&&.01&.02&.01&.01&.01&.01\\
I: 3rd $A$ ignored &.02&.02&.02&.02&.03&.03&&.03&.03&.03&.03&.02&.03\\
II: $h=1$ &.02&.03&.02&.02&.02&.02&&.02&.02&.03&.03&.02&.03\\
II: $h=2$ &.04&.05&.03&.04&.03&.04&&.03&.04&.04&.06&.04&.05\\
III &.05&.06&.05&.06&.05&.06&&.05&.07&.06&.07&.06&.07\\
IV &.04&.05&.03&.04&.03&.04&&.03&.06&.05&.06&.04&.05\\
$U$ ignored &.06&.06&.07&.07&.06&.06&&.06&.06&.07&.07&.06&.06\\\hline
\end{tabular}
\end{table}

\setlength{\tabcolsep}{2pt} 
\begin{table}[H]
\centering
\caption{Average absolute bias (AAB) and root-mean-squared error (RMSE) in the estimated conditional average treatment effects for contextualized simulation in Section 3.2. Three sensitivity analysis strategies were used: 1) true $c^0$, 2) $\mathcal{U} \lp  c^0 - 2\hat{\sigma}, c^0+2\hat{\sigma} \rp$ but bounded within $[-1,1]$,  3) $\mathcal{U}(-1,1)$. The CATE results that could be achieved if $U$ were actually observed and the naive CATE estimators ignoring $U$  are also presented. The true $CATE_{1,2}=0.05$, $CATE_{1,3}=-0.11$ and $CATE_{2,3}=-0.16$.  }
\small
\begin{tabular}{ccccccccccccccc}\hline
&&\multicolumn{6}{c}{$N=1500$, ratio of unit = 1:1:1}&&\multicolumn{6}{c}{$N=10000$, ratio of unit = 1:10:9}\\\cline{3-15}
&&\multicolumn{2}{c}{$CATE_{1,2}$}&\multicolumn{2}{c}{$CATE_{1,3}$}&\multicolumn{2}{c}{$CATE_{2,3}$}&&\multicolumn{2}{c}{$CATE_{1,2}$}&\multicolumn{2}{c}{$CATE_{1,3}$}&\multicolumn{2}{c}{$CATE_{2,3}$}\\
&& AAB&RMSE&AAB&RMSE&AAB&RMSE&&AAB&RMSE&AAB&RMSE&AAB&RMSE\\\hline
&$U$ included&.01&.01&.01&.01&.01&.01&&.00&.00&.00&.00&.00&.00\\
&True $c^0$ &.01&.01&.01&.01&.01&.01&&.00&.00&.00&.00&.00&.00\\
UMC(i) &$\mathcal{U}(c^0-2\hat{\sigma},c^0+2\hat{\sigma})$&.01&.01&.01&.01&.01&.01&&.01&.01&.01&.01&.01&.01\\
&$\mathcal{U}(-1,1)$&.01&.01&.01&.01&.01&.02&&.01&.01&.01&.01&.01&.01\\
&$U$ ignored&.02&.02&.03&.04&.03&.03&&.01&.01&.01&.01&.01&.01\\\hline
&$U$ included&.01&.01&.01&.01&.01&.01&&.00&.00&.00&.00&.00&.00\\
&True $c^0$ &.01&.01&.01&.01&.01&.01&&.00&.00&.00&.00&.00&.00\\
UMC(ii) &$\mathcal{U}(c^0-2\hat{\sigma},c^0+2\hat{\sigma})$&.01&.02&.01&.01&.01&.02&&.01&.01&.01&.01&.01&.01\\
&$\mathcal{U}(-1,1)$&.02&.02&.01&.02&.02&.02&&.01&.01&.01&.01&.01&.02\\
&$U$ ignored&.03&.03&.04&.04&.04&.04&&.02&.02&.02&.02&.02&.02\\\hline
&$U$ included&.01&.01&.01&.01&.01&.01&&.00&.00&.00&.00&.00&.00\\
&True $c^0$ &.01&.01&.01&.01&.01&.01&&.00&.00&.00&.00&.00&.00\\
UMC(iii) &$\mathcal{U}(c^0-2\hat{\sigma},c^0+2\hat{\sigma})$&.02&.03&.02&.03&.02&.03&&.01&.01&.01&.02&.01&.02\\
&$\mathcal{U}(-1,1)$&.03&.03&.02&.03&.03&.04&&.02&.02&.01&.02&.02&.02\\
&$U$ ignored&.04&.04&.05&.05&.05&.05&&.02&.02&.03&.03&.03&.03\\\hline
\end{tabular}
\end{table}

\setlength{\tabcolsep}{5pt} 
\setlength{\tabcolsep}{5pt} 
\begin{table}[H]
\centering
\caption{Sensitivity analysis for causal inference about average treatment effects of three surgical approaches on prolonged length of stay (LOS) based on the risk difference, using the SEER-Medicare lung cancer data. Three surgical approaches are $A=1$: robotic-assisted surgery (RAS), $A=2$: open thoracotomy (OT), $A=3$: video-assisted thoracic surgery (VATS).  The adjusted effect estimates and 95\% uncertainty intervals are displayed. Interval estimates are based on pooled posterior samples across model fits arising from $30 \times 30$ data sets.  The remaining standard deviation in the outcome not explained by the measured covariates is $\hat{\sigma}$ = 0.27 . We assume  $c(1,3,x)  \sim \mathcal{U} (-0.4,0), \; c(3,1,x) \sim \mathcal{U}(0, 0.4), \; c(2,3,x) \sim \mathcal{U} (-0.4,0), \; c(3,2,x) \sim \mathcal{U}(0, 0.4)$ for specification (i)-(v). }
\label{tab:SA-resp}
\begin{tabular}{cp{0.42\textwidth}ccc}
&Prior distributions on $c(\cdot)$ functions  & RAS vs. OT & RAS vs. VATS & OT vs. VATS \\\hline
(i)&$c(1,2,x) \sim \mathcal{U} (-0.2,0),  c(2,1,x) \sim \mathcal{U}(0, 0.2)$ &  $-.05 (-.08,-.02)$ & $.06(.03,.09)$ & $.10 (.08,.12)$\\
(ii)&$ c(1,2,x) \sim \mathcal{U}(0, 0.2), c(2,1,x) \sim \mathcal{U} (-0.2,0)$ &  $-.11 (-.14,-.08)$ & $.05(.02,.08)$ & $.12 (.10,.14)$\\
(iii)&$c(1,2,x), c(2,1,x) \sim \mathcal{U} (-0.2,0) $ & $-.05(-.08,-.02)$ & $.04(.01,.07)$ & $.14(.12,.16)$ \\
(iv)&$c(1,2,x), c(2,1,x) \sim \mathcal{U}(0, 0.2)$& $-.11(-.14,-.08)$ & $.00(-.03,.03)$ & $.08 (.06,.10)$\\
(v)&$c(1,2,x), c(2,1,x) \sim \mathcal{U}(-0.2, 0.2)$& $-.10(-.14,-.06)$ & $.01(-.03,.05)$ & $.09 (.06,.12)$\\
(vi)&all $c(\cdot) \sim \mathcal{U}(-1, 1)$ & $-.07(-.15,.01)$ & $.05(-.03,.13)$ & $.12(.06,.18)$ \\
(vii)&all $c(\cdot) =0$ & $-.09(-.11,-.07)$ & $.02(-.00,.04)$ & $.11(.10,.12)$\\\hline
\end{tabular}
\end{table}

\begin{table}[H]
\centering
\caption{Sensitivity analysis for causal inference about average treatment effects of three surgical approaches on intensive care unit (ICU) stay based on the risk difference, using the SEER-Medicare lung cancer data. Three surgical approaches are $A=1$: robotic-assisted surgery (RAS), $A=2$: open thoracotomy (OT), $A=3$: video-assisted thoracic surgery (VATS). The adjusted effect estimates and 95\% uncertainty intervals are displayed. Interval estimates are based on pooled posterior samples across model fits arising from $30 \times 30$ data sets.  The remaining standard deviation in the outcome not explained by the measured covariates is $\hat{\sigma}$ = 0.46 . We assume  $c(1,3,x)  \sim \mathcal{U} (-0.4,0), \; c(3,1,x) \sim \mathcal{U}(0, 0.4), \; c(2,3,x) \sim \mathcal{U} (-0.4,0), \; c(3,2,x) \sim \mathcal{U}(0, 0.4)$ for specification (i)-(v).}
\label{tab:SA-resp}
\begin{tabular}{cp{0.42\textwidth}ccc}
&Prior distributions on $c(\cdot)$ functions  & RAS vs. OT & RAS vs. VATS & OT vs. VATS \\\hline
(i)&$c(1,2,x) \sim \mathcal{U} (-0.4,0),  c(2,1,x) \sim \mathcal{U}(0, 0.4)$ &  $-.11 (-.14,-.07)$ & $.02(-.01,.05)$ & $.17 (.15,.19)$\\
(ii)&$ c(1,2,x) \sim \mathcal{U}(0, 0.4),c(2,1,x) \sim \mathcal{U} (-0.4,0)$ &  $-.15 (-.18,-.12)$ & $.00(-.03,.03)$ & $.15 (.13,.17)$\\
(iii)&$c(1,2,x), c(2,1,x) \sim \mathcal{U} (-0.4,0) $ & $-.13(-.16,-.10)$ & $.03(-.00,.06)$ & $.16(.14,.18)$ \\
(iv)&$c(1,2,x), c(2,1,x) \sim \mathcal{U}(0, 0.4)$ & $-.15(-.18,-.12)$ & $-.02(-.05,.01)$ & $.11 (.09,.13)$\\
(v)&$c(1,2,x), c(2,1,x) \sim \mathcal{U}(-0.4, 0.4)$ & $-.14(-.18,-.10)$ & $-.01(-.05,.03)$ & $.10 (.07,.13)$\\
(vi)&all $c(\cdot) \sim \mathcal{U}(-1, 1)$ & $-.13(-.21,-.05)$ & $.03(-.05,.11)$ & $.16(.10,.24)$ \\
(vii)&all $c(\cdot) =0$  & $-.14(-.16,-.12)$ & $.01(-.00,.02)$ & $.16(.15,.17)$\\\hline
\end{tabular}
\end{table}

\begin{table}[H]
\centering
\caption{Sensitivity analysis for causal inference about average treatment effects of three surgical approaches on 30-day readmission rate based on the risk difference, using the SEER-Medicare lung cancer data. Three surgical approaches are $A=1$: robotic-assisted surgery (RAS), $A=2$: open thoracotomy (OT), $A=3$: video-assisted thoracic surgery (VATS).  The adjusted effect estimates and 95\% uncertainty intervals are displayed. Interval estimates are based on pooled posterior samples across model fits arising from $30 \times 30$ data sets.  The remaining standard deviation in the outcome not explained by the measured covariates is $\hat{\sigma}$ = 0.28 . We assume  $c(1,3,x)  \sim \mathcal{U} (-0.4,0), \; c(3,1,x) \sim \mathcal{U}(0, 0.4), \; c(2,3,x) \sim \mathcal{U} (-0.4,0), \; c(3,2,x) \sim \mathcal{U}(0, 0.4)$ for specification (i)-(v).}
\label{tab:SA-resp}
\begin{tabular}{cp{0.42\textwidth}ccc}
&Prior distributions on $c(\cdot)$ functions  & RAS vs. OT & RAS vs. VATS & OT vs. VATS \\\hline
(i)&$c(1,2,x) \sim \mathcal{U} (-0.2,0),  c(2,1,x) \sim \mathcal{U}(0, 0.2)$ &  $.02 (-.01,.05)$ & $.06(.03,.09)$ & $.03 (.01,.05)$\\
(ii)&$c(1,2,x) \sim \mathcal{U}(0, 0.2), c(2,1,x) \sim \mathcal{U} (-0.2,0)$ &  $-.01 (-.04,.02)$ & $.03(-.00,.06)$ & $.05(.03,.07)$\\
(iii)&$c(1,2,x), c(2,1,x) \sim \mathcal{U} (-0.2,0) $ & $.03(-.00,.06)$ & $.05(.02,.08)$ & $.07(.05,.09)$ \\
(iv)&$c(1,2,x), c(2,1,x) \sim \mathcal{U}(0, 0.2)$ & $-.01(-.04,.02)$ & $.01(-.02,.04)$ & $.00 (-.02,.02)$\\
(v)&$c(1,2,x), c(2,1,x) \sim \mathcal{U}(-0.2, 0.2)$ & $-.00(-.04,.04)$ & $.02(-.02,.06)$ & $.01 (-.03,.05)$\\
(vi)&all $c(\cdot) \sim \mathcal{U}(-1, 1)$ & $.01(-.07,.09)$ & $.05(-.03,.13)$ & $.04(-.03,.11)$ \\
(vii)&all $c(\cdot) =0$ & $-.00(-.02,.02)$ & $.02(-.00,.04)$ & $.02(-.00,.04)$\\\hline
\end{tabular}
\end{table}

\begin{table}[H]
\centering
\caption{Sensitivity analysis for causal inference about average treatment effects among those who were operated with robotic-assisted surgery  on postoperative respiratory complications based on the risk difference, using the SEER-Medicare lung cancer data. Three surgical approaches are $A=1$: robotic-assisted surgery (RAS), $A=2$: open thoracotomy (OT), $A=3$: video-assisted thoracic surgery (VATS).  The adjusted estimates of causal effects and 95\% uncertainty intervals are displayed. Interval estimates are based on pooled posterior samples across model fits arising from $30 \times 30$ data sets.  The remaining standard deviation in the outcome not explained by the measured covariates is $\hat{\sigma}$ = 0.46. We assume  $c(1,3,x)  \sim \mathcal{U} (-0.4,0), \; c(3,1,x) \sim \mathcal{U}(0, 0.4), \; c(2,3,x) \sim \mathcal{U} (-0.4,0), \; c(3,2,x) \sim \mathcal{U}(0, 0.4)$ for specification (i)-(v).}
\label{tab:SA-resp}
\begin{tabular}{cp{0.42\textwidth}cc} \\
&Prior distributions on $c(\cdot)$ functions  & RAS vs. OT & RAS vs. VATS  \\\hline
(i)&$c(1,2,x) \sim \mathcal{U} (-0.4,0),  c(2,1,x) \sim \mathcal{U}(0, 0.4)$ &  $.02 (-.02,.06)$   & $.04(.00,.08)$ \\
(ii)&$ c(1,2,x) \sim \mathcal{U}(0, 0.4), c(2,1,x) \sim \mathcal{U} (-0.4,0)$ &  $-.02 (-.06,.02)$& $.02(-.02,.06)$  \\
(iii) & $c(1,2,x), c(2,1,x) \sim \mathcal{U} (-0.4,0) $ & $.03(-.01,.07)$ & $.04(.00,.08)$ \\
(iv) & $c(1,2,x), c(2,1,x) \sim \mathcal{U}(0, 0.4)$ & $.01(-.03,.05)$ & $.01(-.03,.05)$  \\
(v)&$c(1,2,x), c(2,1,x) \sim \mathcal{U}(-0.4, 0.4)$ & $.02(-.03,.07)$ & $.02(-.03,.07)$  \\
(vi)&all $c(\cdot) \sim \mathcal{U}(-1, 1)$ & $.01(-.09,.11)$ & $.03(-.07,.13)$  \\
(vii) &all $c(\cdot) =0$ & $-.01(-.04,.02)$ & $.00(-.03,.03)$ \\\hline
\end{tabular}
\end{table}

\begin{figure}[H] 
\includegraphics[width=0.9\textwidth]{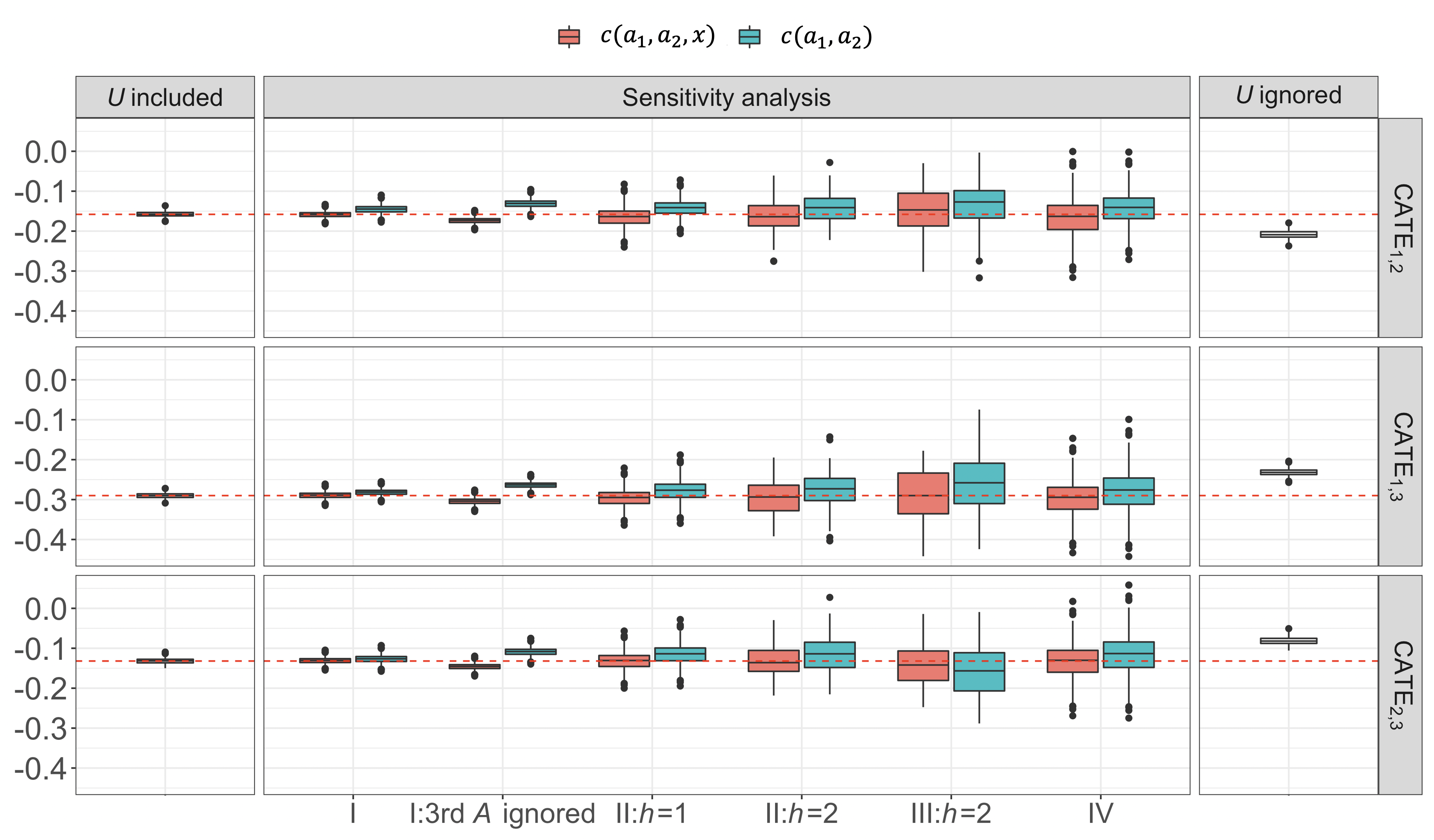}
\caption{Estimates of three pairwise causal effects $\text{CATE}_{1,2}$, $\text{CATE}_{1,3}$ and $\text{CATE}_{2,3}$ among 1000 replications. Two independent confounders, measured $X_1$ and unmeasured $U$, are considered. The data generating models are described in Web Section 1. For sensitivity analysis, strategies (I)-(IV) described in Section 3.1 were used to specify the prior distributions for the confounding functions $c(\cdot)$. For strategy (I),  the scenario ``3rd $A$ ignored'' considers only the $c(\cdot)$ functions involving the pair of treatments for the target CATE, while setting the  $c(\cdot)$ functions involving the third treatment to zero. For strategy (II), both $h=1$ and $h=2$ are considered, representing one and two remaining standard deviation, respectively. The CATE results that could be achieved if $U$ were actually observed and the naive CATE estimators ignoring $U$  are also presented. Red dashed lines mark the true CATE.}
\label{fig:sim_binary_no_interaction}
\end{figure}

\begin{figure}[H] 
\includegraphics[width=0.9\textwidth]{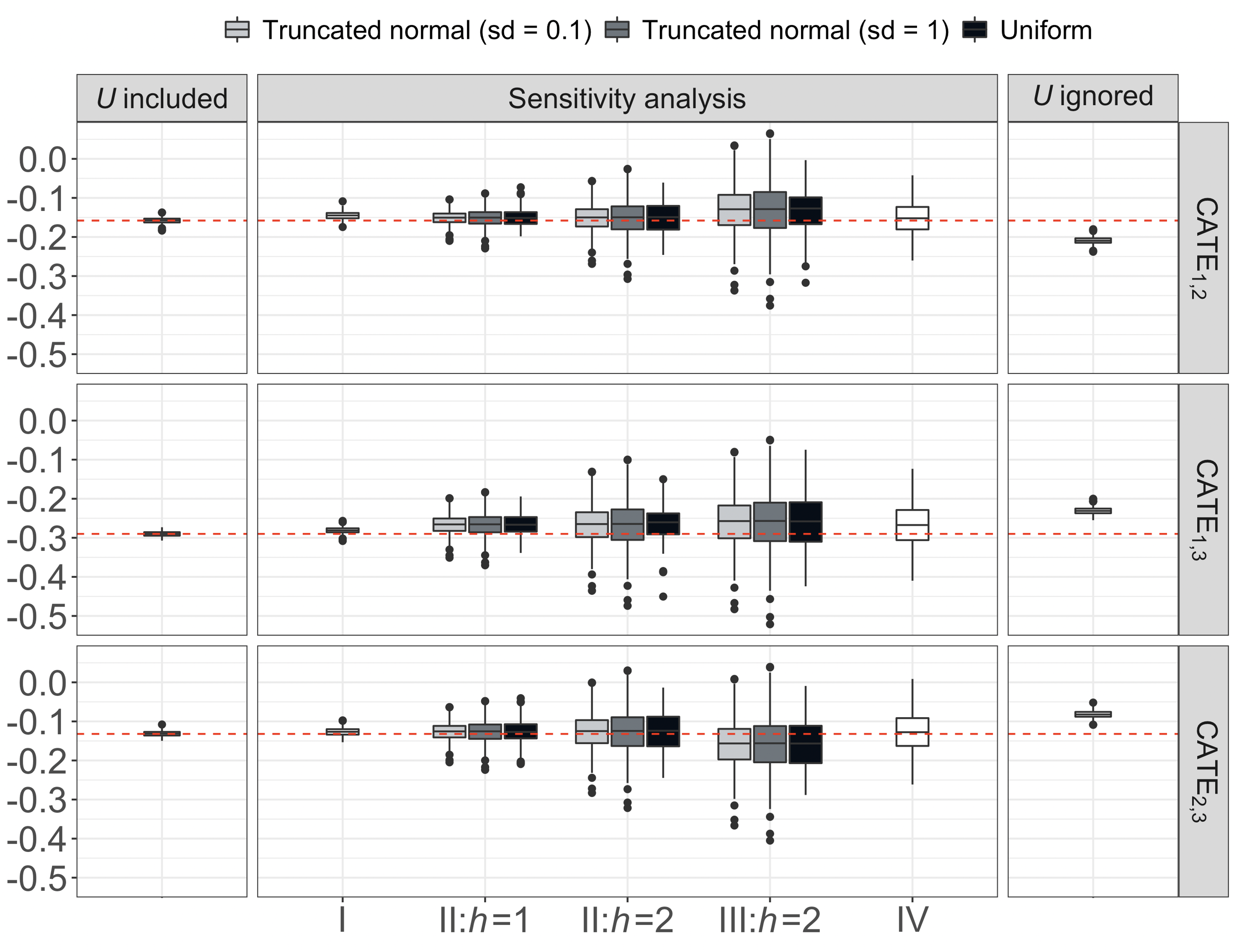}
\caption{Estimates of three pairwise causal treatment effects $CATE_{1,2}$, $CATE_{1,3}$ and $CATE_{2,3}$ among 1000 replications. For sensitivity analysis, strategies (I)-(IV) described in Section 3.1 were used to specify the prior distributions for the confounding functions $c(\cdot)$. For strategy (II), both $h=1$ and $h=2$ are considered, representing one and two remaining standard deviation, respectively. For strategies (II) and (III), both the uniform distribution and the truncated normal distribution with large and small spread were used. The CATE results that could be achieved if $U$ were actually observed and the naive CATE estimators ignoring $U$  are also presented.  Red dashed lines mark the true CATE.}
\label{fig:sim_binary_suppl}
\end{figure}

\begin{figure}[H] 
\includegraphics[width=1\textwidth]{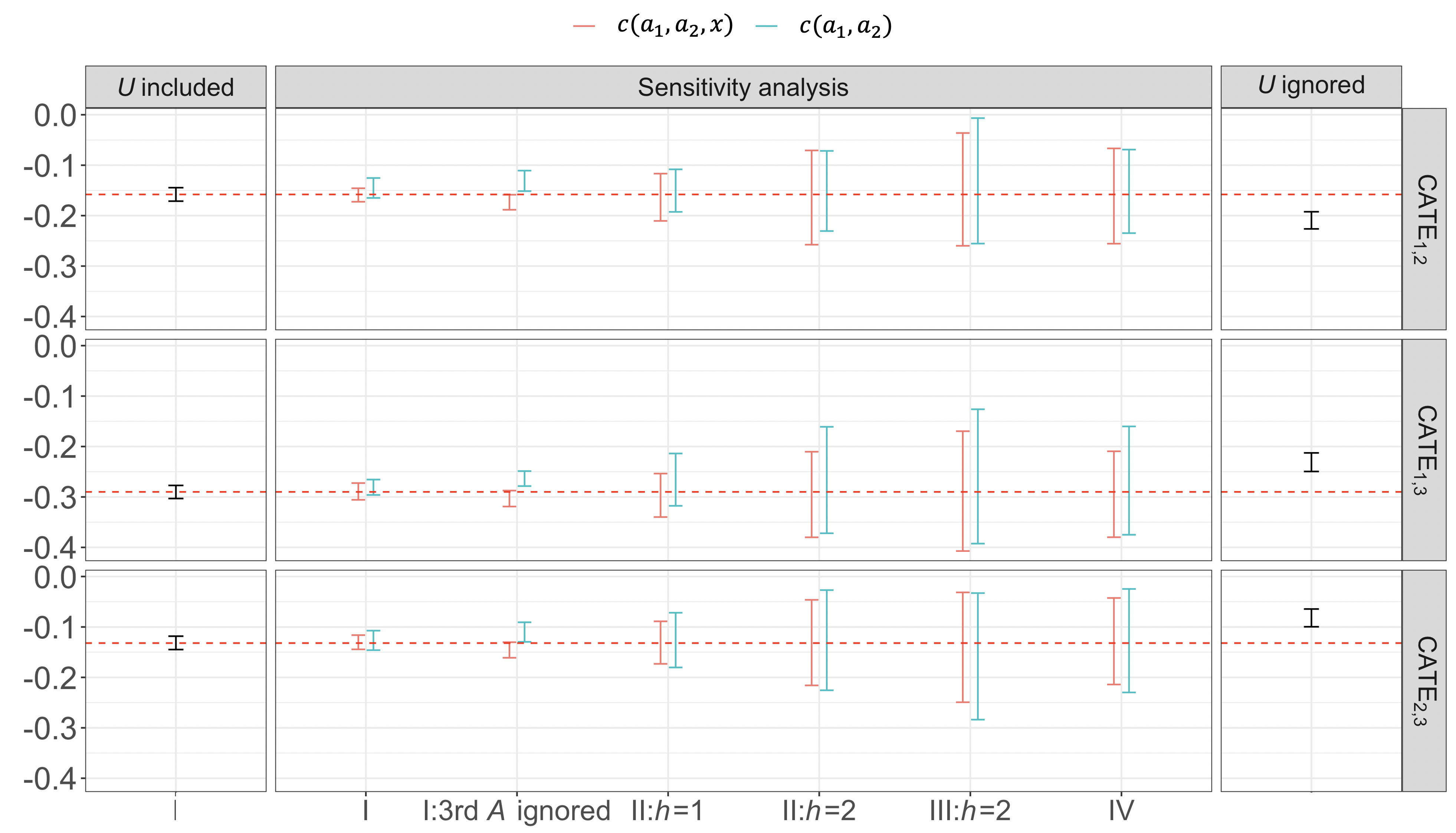}
\caption{Credible intervals for three pairwise causal effects $\text{CATE}_{1,2}$, $\text{CATE}_{1,3}$ and $\text{CATE}_{2,3}$ using a random data replication. For sensitivity analysis, strategies (I)-(IV) described in Section 3.1 were used to specify the prior distributions for the confounding functions $c(\cdot)$. For strategy (II), both $h=1$ and $h=2$ are considered, representing one and two remaining standard deviation, respectively. For strategies (II) and (III), both the uniform distribution and the truncated normal distribution with large and small spread were used. The CATE results that could be achieved if $U$ were actually observed and the naive CATE estimators ignoring $U$  are also presented.  Red dashed lines mark the true CATE.Red dashed lines mark the true CATE.}
\label{fig:sim_binary_ci}
\end{figure}

\begin{figure}[H] 
\includegraphics[width=1\textwidth]{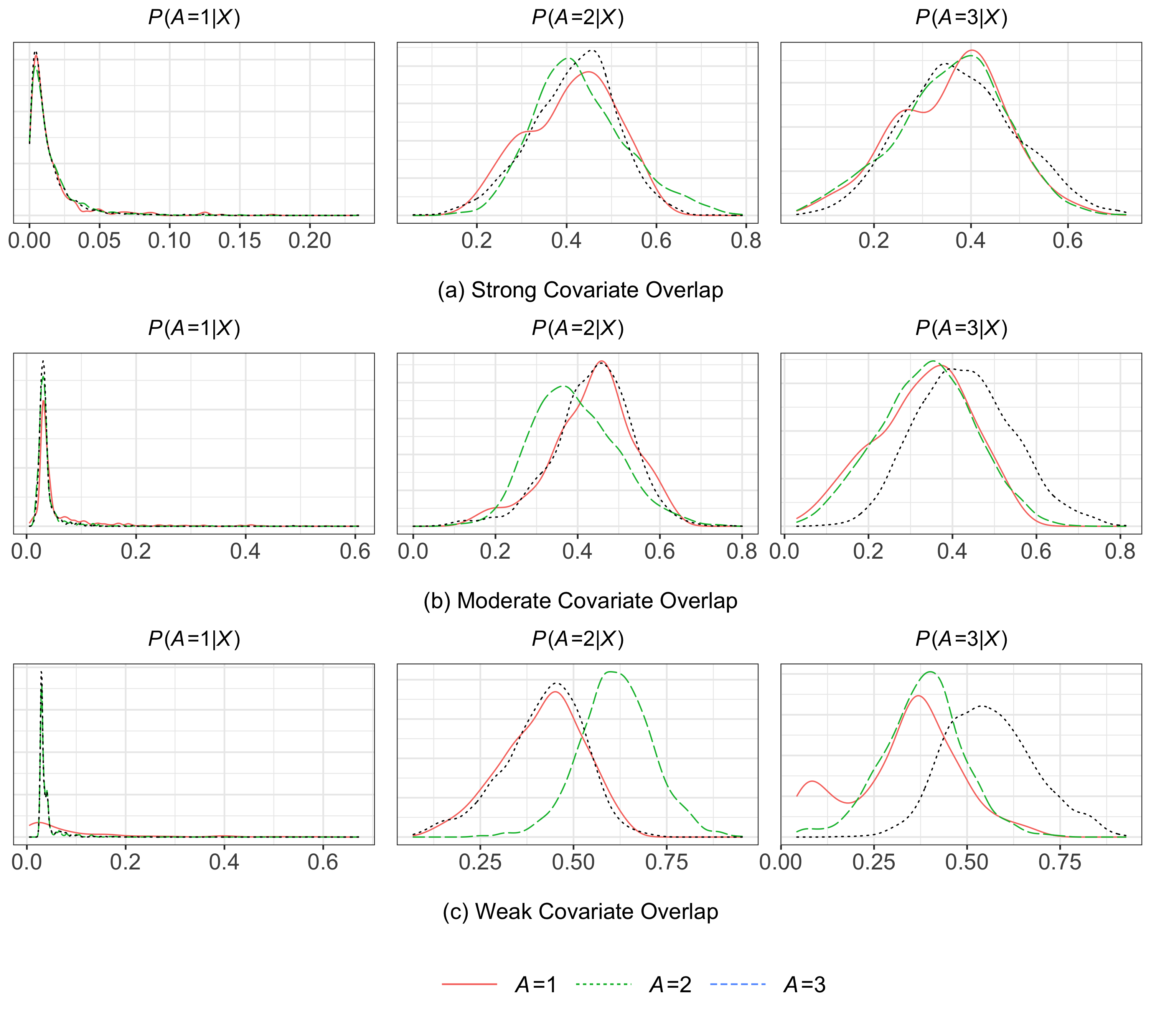}
\caption{Distributions of the true generalized propensity scores (GPS) corresponding to (a) strong, (b) moderate, and (c) weak covariate overlap. Each panel presents the density plots of the GPS for units assigned to a given treatment group. The left panel corresponds to treatment 1, the middle panel treatment 2, and the right panel  treatment 3.} 
\label{fig:overlap_sim}
\end{figure}

\begin{figure}[H] 
\includegraphics[width=1\textwidth]{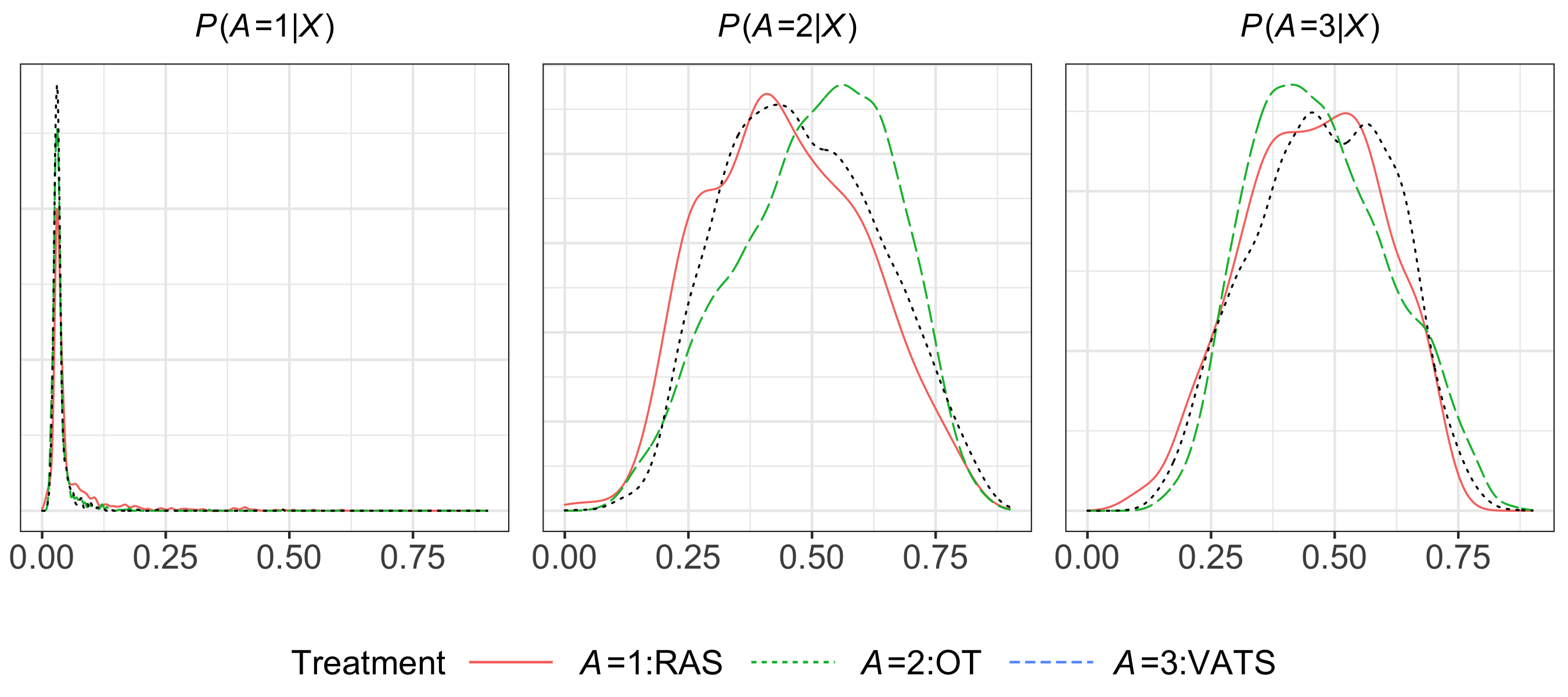}
\caption{Distributions of the posterior mean generalized propensity scores (GPS) for the SEER-Medicare lung cancer data. The GPS were obtained by fitting a BART model for the multinomial outcomes. Each panel presents the density plots of the GPS for units assigned to a given treatment group. The left panel corresponds to treatment 1, the middle panel treatment 2, and the right panel  treatment 3.} 
\end{figure}

\begin{figure}[H] 
\includegraphics[width=0.9\textwidth]{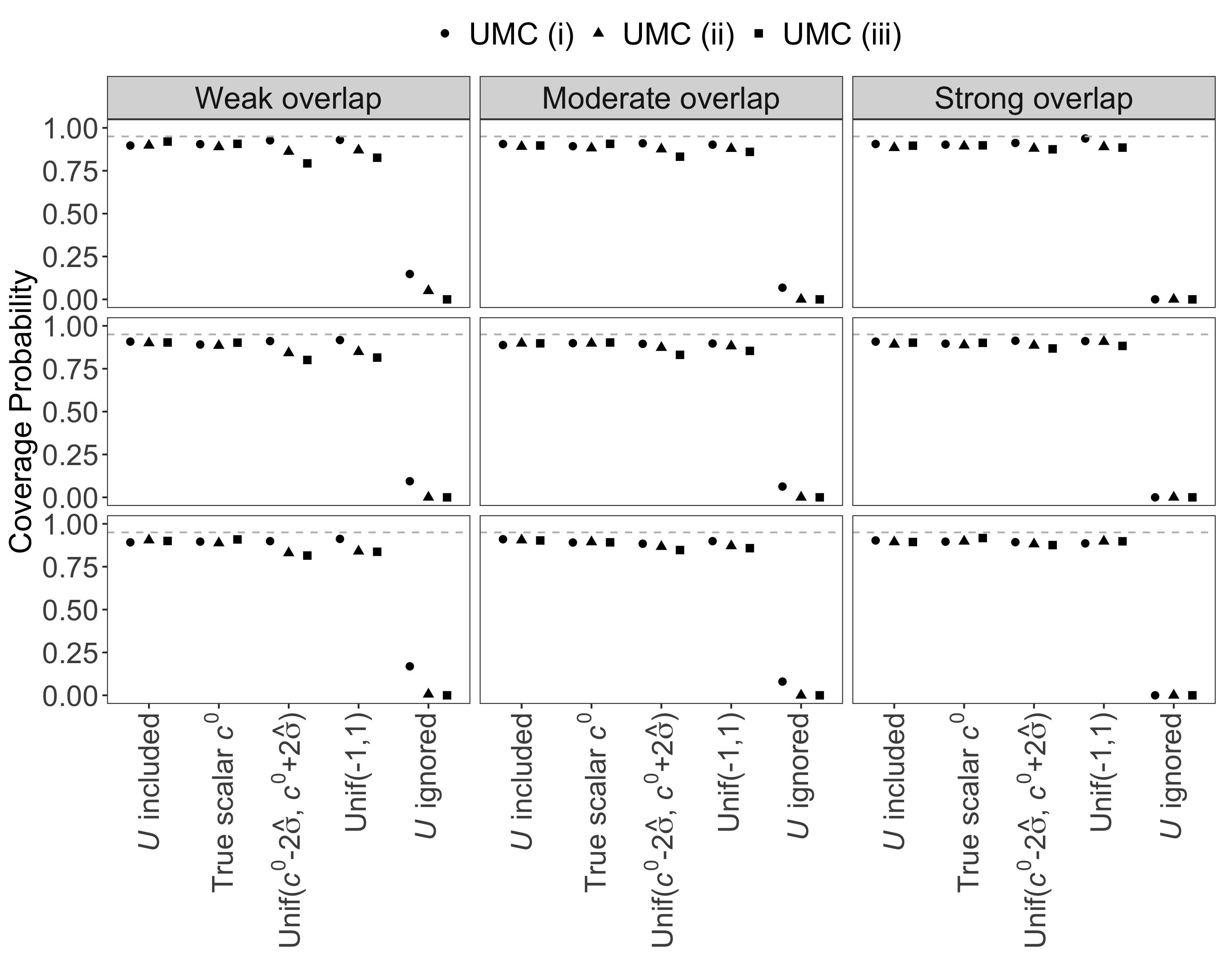}
\caption{The coverage probability of the estimates of three pairwise causal  effects $\text{CATE}_{1,2}$, $\text{CATE}_{1,3}$ and $\text{CATE}_{2,3}$ among 1000 replications, for the sample size $N=1500$ with ratio of units = 1:1:1, three complexity levels of unmeasured confounding, UMC(i), UMC(ii) and UMC(iii), and three levels of covariate overlap, strong, moderate and weak,  as described in Section 3.2.1. For sensitivity analysis, three strategies were used to specify the prior distributions for the confounding functions: 1) true scalar parameters $c^0$, 2) $\mathcal{U}(c^0-2\hat{\sigma}, c^0+2\hat{\sigma})$ bounded within $[-1,1]$, and 3) $\mathcal{U}(-1,1)$. Coverage probabilities of the CATE results that could be achieved if $U$ were actually observed and the naive CATE estimators ignoring $U$  are also presented. Gray dashed lines mark the nominal 95\% coverage probability.} 
\label{fig: sim-complex-CP-1500}
\end{figure}

\begin{figure}[H] 
\centering
\includegraphics[scale=0.45]{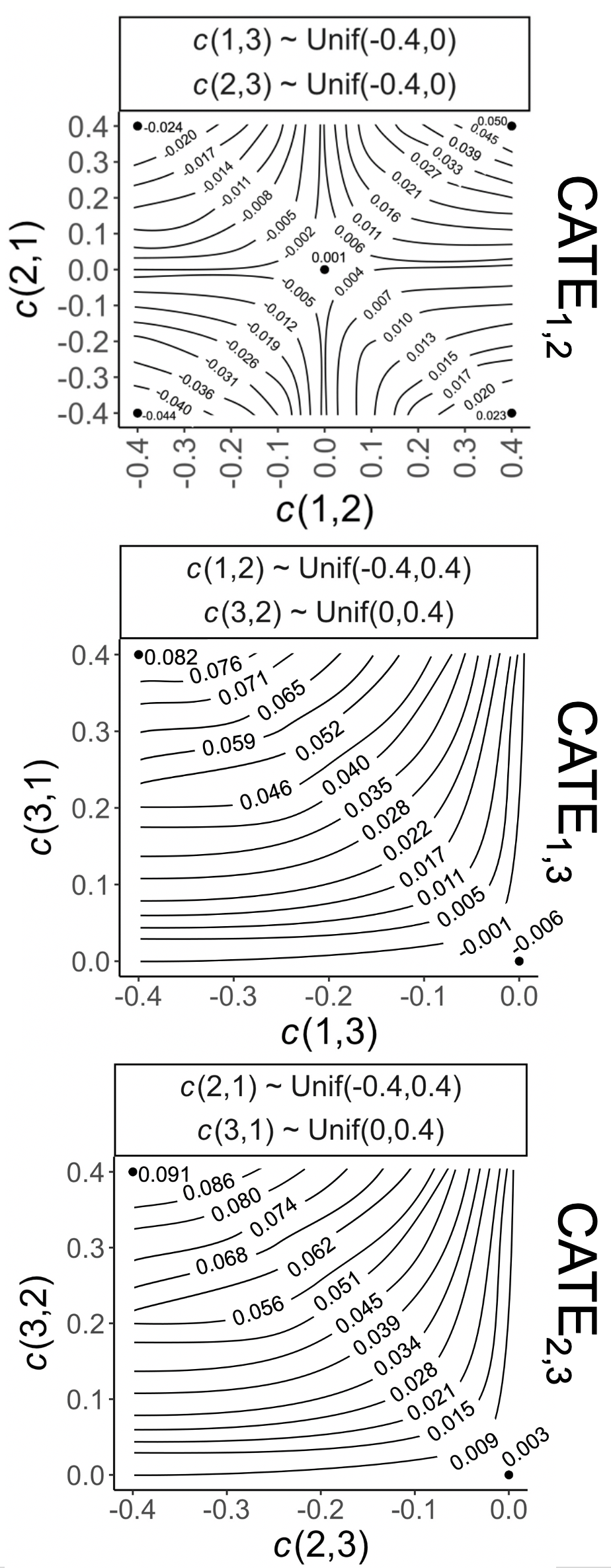}
\caption{Contour plots of the confounding function adjusted treatment effect estimates for RAS versus OT, $CATE_{1,2}$, RAS versus VATS, $CATE_{1,3}$ and OT vs VATS, $CATE_{2,3}$. Four confounding functions are involved in each pairwise treatment effect.  The  black lines in each panel report the adjusted causal effect estimates of $a_j$ versus $a_k$, $\forall j <k \in \{1,2,3\}$ under different pairs of values for  $c(a_j,a_k,x)$ and $c(a_k,a_j,x)$ spaced on a grid, and under the prior distributions for the other two confounding functions. } 
\end{figure}

\section{Sample codes} 
We provide step-by-step sample coding of the illustrative simulation described in Section 3.1 for the scenario where we use the true $c^0$. The proposed method is readily available in the $\R$ package  $\textsf{SAMTx}$. 

\begin{lstlisting}[breaklines=true]
# Simulate the data
library(BART)
sample_size = 1500
# First simulate the treatment indicator A
x1 = rbinom(sample_size, 1, prob=0.4)
u = rbinom(sample_size, 1, prob=0.5)
lp.A = 0.2 * x1 + 0.4 * u + rnorm(sample_size, 0, 0.1)
lp.B = -0.3 * x1 + 0.8 * u + rnorm(sample_size, 0, 0.1)
lp.C = 0.1 * x1 + 0.5 * u + rnorm(sample_size, 0, 0.1)
p.A1 <- exp(lp.A)/(exp(lp.A)+exp(lp.B)+exp(lp.C))
p.A2 <- exp(lp.B)/(exp(lp.A)+exp(lp.B)+exp(lp.C))
p.A3 <- exp(lp.C)/(exp(lp.A)+exp(lp.B)+exp(lp.C))
p.A <- matrix(c(p.A1,p.A2,p.A3),ncol = 3)
A = NULL
for (i in 1:sample_size) { 
  A[i] <- sample(c(1, 2, 3),
                 size = 1,
                 replace = TRUE,
                 prob = p.A[i, ])
}
table(A)
# Then simulate the treatment P(Y(A) = 1|x1, u)
Y1 = -0.8 * x1 - 1.2 * u + 1.5 * u * x1
Y2 = -0.6 * x1 + 0.5 * u + 0.3 * x1 * u
Y3 = 0.3 * x1 + 0.2 * x1 * u + 1.3 * u
Y1 = rbinom(sample_size, 1, exp(Y1)/(1+exp(Y1)))
Y2 = rbinom(sample_size, 1, exp(Y2)/(1+exp(Y2)))
Y3 = rbinom(sample_size, 1, exp(Y3)/(1+exp(Y3)))
dat_truth = cbind(Y1, Y2, Y3, A) # True data for the outcome Y
Yobs <- apply(dat_truth, 1, function(x) # Observed data for the outcome Y
  x[1:3][x[4]]) 
# Simulate the true confounding function c(a1, a2)
n_alpha = 30
alpha = cbind(
  runif(n_alpha, mean(Y1[A ==1])-mean(Y1[A ==2]) - 0.001, mean(Y1[A ==1])-mean(Y1[A ==2]) + 0.001),
  runif(n_alpha, mean(Y2[A ==2])-mean(Y2[A ==1]) - 0.001, mean(Y2[A ==2])-mean(Y2[A ==1]) + 0.001),
  runif(n_alpha, mean(Y2[A ==2])-mean(Y2[A ==3]) - 0.001, mean(Y2[A ==2])-mean(Y2[A ==3]) + 0.001),
  runif(n_alpha, mean(Y1[A ==1])-mean(Y1[A ==3]) - 0.001, mean(Y1[A ==1])-mean(Y1[A ==3]) + 0.001),
  runif(n_alpha, mean(Y3[A ==3])-mean(Y3[A ==1]) - 0.001, mean(Y3[A ==3])-mean(Y3[A ==1]) + 0.001),
  runif(n_alpha, mean(Y3[A ==3])-mean(Y3[A ==2]) - 0.001, mean(Y3[A ==3])-mean(Y3[A ==2]) + 0.001))
y <- Yobs
covariates = as.matrix(cbind(x1, u))
y = as.numeric(y)
A = as.integer(A)
A_unique_length <- length(unique(A))
alpha = as.matrix(alpha)
M1 <- 30 # This is M1
M2 <- 30 # This is M2
nposterior <- 10000 # 10000 posterior samples

# Algorithm 1.1: Fit the multinomial probit BART model to the treatment A 
A_model = mbart2(x.train = covariates, as.integer(as.factor(A)), x.test = covariates, ndpost = nposterior)

# Algorithm 1.1: Estimate the generalized propensity scores for each individual 
p = array(A_model$prob.test[seq(1, nrow(A_model$prob.test), nposterior),], dim = c(M1, 3, length(A)))

# Algorithm 1.2: start to calculate causal effect by M1 * M2 times 
causal_effect_1 = matrix(NA, nrow = M2 * M1, ncol = nposterior)
causal_effect_2 = matrix(NA, nrow = M2 * M1, ncol = nposterior)
causal_effect_3 = matrix(NA, nrow = M2 * M1, ncol = nposterior)
step = 1
train_x = cbind(covariates, A)
for (j in 1:M1) {
  # Algorithm 1.2: Draw M1 generalzied propensity scores from the posterior predictive distribution of the A model for each individual 
  p_draw_1 <- p[j, 1, ]
  p_draw_2 <- p[j, 2, ]
  p_draw_3 <- p[j, 3, ]
  for (m in 1:M2) {
    # Algorithm 1.2: Draw M2 values from the prior distribution of each of the sensitivity paramaters alpha for eacg treatment 
    print(paste("step :", step, "/", M2*M1))
    sort(unique(train_x[, "A"]))
    # Algorithm 1.3: Compute the adjusted outcomes y_corrected for each treatment for each M1M2 draws
    y_corrected = ifelse(
      train_x[, "A"] == sort(unique(train_x[, "A"]))[1],
      y - (unlist(alpha[m, 1]) * p_draw_2  + unlist(alpha[m, 4]) * p_draw_3),
      ifelse(
        train_x[, "A"] == sort(unique(train_x[, "A"]))[2],
        y - (unlist(alpha[m, 2]) * p_draw_1 + unlist(alpha[m, 3]) * p_draw_3),
        y - (unlist(alpha[m, 5]) * p_draw_1 + unlist(alpha[m, 6]) * p_draw_2)
      )
    )
    # Algorithm 1.4: Fit a BART model to each set of M1*M2 sets of observed data with the adjusted outcomes y_corrected
    bart_mod = wbart(x.train = cbind(covariates, A),  y.train = y_corrected,  ndpost = nposterior, printevery = 10000)
    
    predict_1 = pwbart(cbind(covariates, A = sort(unique(A))[1]), bart_mod$treedraws)
    predict_2 = pwbart(cbind(covariates, A = sort(unique(A))[2]), bart_mod$treedraws)
    predict_3 = pwbart(cbind(covariates, A = sort(unique(A))[3]), bart_mod$treedraws)
    causal_effect_1[((m - 1) * M1) + j, ] = rowMeans(predict_1 - predict_2)
    causal_effect_2[((m - 1) * M1) + j, ] = rowMeans(predict_2 - predict_3)
    causal_effect_3[((m - 1) * M1) + j, ] = rowMeans(predict_1 - predict_3)
    step = step + 1
  }
  
}
# Final combined adjusted causal effect
ATE_01_adjusted <- causal_effect_1
ATE_12_adjusted <- causal_effect_2
ATE_02_adjusted <- causal_effect_3
\end{lstlisting}

